\documentclass[moor]{informs-arxiv}              % for a regular run
\usepackage{amsmath,amssymb,amsfonts}
\usepackage{bbm}
\usepackage{mathtools}

\usepackage{natbib}
 \NatBibNumeric
 \def\bibfont{\small}%
 \bibpunct[, ]{[}{]}{,}{n}{}{,}%
\usepackage{acronym}
\usepackage[colorlinks=true,breaklinks=true,bookmarks=true,urlcolor=blue,
     citecolor=blue,linkcolor=blue,bookmarksopen=false,draft=false]{hyperref}

\TheoremsNumberedThrough     % Preferred (Theorem 1, Lemma 1, Theorem 2)
\EquationsNumberedThrough    % Default: (1), (2), ...
\usepackage[T1]{fontenc}

\usepackage{dsfont}

\usepackage[mathscr]{eucal}

\usepackage{acronym}
\usepackage[shortlabels]{enumitem}
\usepackage{latexsym}
\usepackage{wasysym}
\usepackage{xspace}
\usepackage{mathscinet}
\usepackage{longtable}

\usepackage[dvipsnames,svgnames]{xcolor}
\colorlet{MyBlue}{DodgerBlue!60!Black}
\colorlet{MyGreen}{DarkGreen!85!Black}

\usepackage{graphicx}
\usepackage[left=1.4in,right=1.4in]{geometry}
\usepackage[font=small,labelfont=bf]{caption}
\usepackage{subfigure}
\usepackage{tikz}
\usetikzlibrary{calc,patterns}

\usepackage{url}
\usepackage{hyperref}
\hypersetup{
	colorlinks=true,
	linktocpage=true,
	pdfstartview=FitH,
	breaklinks=true,
	pdfpagemode=UseNone,
	pageanchor=true,
	pdfpagemode=UseOutlines,
	plainpages=false,
	bookmarksnumbered,
	bookmarksopen=false,
	bookmarksopenlevel=1,
	hypertexnames=true,
	pdfhighlight=/O,
	urlcolor=black, linkcolor=black, citecolor=black, %pagecolor=black, % <--- for printing
	pdftitle={},
	pdfauthor={},
	pdfsubject={},
	pdfkeywords={},
	pdfcreator={pdfLaTeX},
	pdfproducer={LaTeX with hyperref}
}

\usepackage[sort&compress,capitalize,nameinlink]{cleveref}
\usepackage[textwidth=30mm]{todonotes}
\usepackage{soul}
\setstcolor{red}
\sethlcolor{SkyBlue}

\newcommand{\debug}[1]{#1}

\DeclarePairedDelimiter{\braces}{\{}{\}}
\DeclarePairedDelimiter{\bracks}{[}{]}
\DeclarePairedDelimiter{\parens}{(}{)}
\DeclarePairedDelimiter{\abs}{\lvert}{\rvert}

\DeclarePairedDelimiter{\floor}{\lfloor}{\rfloor}

\DeclarePairedDelimiterX{\braket}[2]{\langle}{\rangle}{#1,#2}
\DeclarePairedDelimiterX{\inner}[2]{\langle}{\rangle}{#1,#2}
\DeclarePairedDelimiterX{\setdef}[2]{\{}{\}}{#1:#2}
\DeclarePairedDelimiterXPP{\probof}[1]{\prob}{(}{)}{}{%
	#1}
\DeclarePairedDelimiterXPP{\exof}[1]{\ex}{[}{]}{}{%
	#1}

\DeclareMathOperator{\bigoh}{\mathscr{O}}

\DeclareMathOperator{\cupdot}{\dot\cup}

\DeclareMathOperator{\ex}{\debug{\mathbb{E}}}
\DeclareMathOperator{\expect}{\debug{\mathbb{E}}}

\DeclareMathOperator{\prob}{\debug{\mathbb{P}}}

\newcommand{\naturals}{\mathbb{\debug N}}
\newcommand{\reals}{\mathbb{\debug R}}

\newcommand{\simplex}{\debug \Sigma}
\newcommand{\run}{\debug k}
\newcommand{\subrun}{\debug m}

\newcommand{\event}{\debug A}
\newcommand{\initial}{\debug \mu}
\newcommand{\transit}{\debug \pi}
\newcommand{\transitvec}{\boldsymbol{\transit}}
\newcommand{\transitvecalt}{\transitvec'}
\newcommand{\transitmatrix}{\transitvec}

\newcommand{\stationary}{\debug \rho}

\newcommand{\proba}{\debug p}

\newcommand{\horizon}{\debug T}
\newcommand{\per}{\debug t}

\newcommand{\graph}{\mathcal{\debug G}}

\newcommand{\vertices}{\mathcal{\debug V}}
\newcommand{\edge}{\debug e}
\newcommand{\edges}{\mathcal{\debug E}}

\newcommand{\hamil}{\mathcal{\debug H}}

\newcommand{\unicycles}{\mathcal{\debug U}}

\newcommand{\tree}{\debug \tau}
\newcommand{\trees}{\mathcal{\debug T}}

\newcommand{\neighborsout}{\mathcal{\debug N}^{+}}

\newcommand{\completegraph}{\debug K}
\newcommand{\cyclegraph}{\debug C}
\newcommand{\ringstargraph}{\debug W}
\newcommand{\treeweight}{\debug \Omega}

\newcommand{\game}{\debug \Gamma}
\newcommand{\bhgame}{\widehat{\game}}

\newcommand{\play}{\debug i}
\newcommand{\playalt}{\debug j}
\newcommand{\playaltalt}{\debug k}
\newcommand{\nPlayers}{\debug n}

\newcommand{\act}{\debug s}
\newcommand{\actalt}{\act'}
\newcommand{\actprof}{\boldsymbol{\act}}
\newcommand{\actprofalt}{\debug{\actprof'}}
\newcommand{\actions}{\mathcal{\debug {S}}}
\newcommand{\sactions}{\debug \Xi}

\newcommand{\cost}{\debug c}

\newcommand{\costs}{\boldsymbol{\cost}}

\newcommand{\potential}{\debug \Psi}

\DeclareMathOperator{\NE}{\mathsf{\debug {NE}}}

\DeclareMathOperator{\PoA}{\mathsf{\debug {PoA}}}

\DeclareMathOperator{\PoS}{\mathsf{\debug {PoS}}}
\DeclareMathOperator{\SC}{\mathsf{\debug {SC}}}

\newcommand{\absorb}{\mathsf{\debug{P}}}
\newcommand{\buck}{\debug \Theta}
\newcommand{\countrc}{\debug{\ell}} %the label of a recurrent class
\newcommand{\countrcalt}{\debug h}
\newcommand{\dumping}{\debug \alpha}
\newcommand{\halfplay}{\debug{\widetilde{m}}}
\newcommand{\indunic}{\debug \delta}
\newcommand{\length}{\debug \Lambda}
\newcommand{\markov}{\debug X}
\newcommand{\numrc}{\debug M} %number of recurrent components
\newcommand{\numroottrees}{\debug R}
\newcommand{\pex}{\debug p}
\newcommand{\qex}{\debug q}
\newcommand{\probalt}{\mathbb{\debug Q}} % probability measure on subset of edges
\newcommand{\PRact}{\debug b} % strategy in the PageRank model
\newcommand{\PRactprof}{\boldsymbol{\PRact}} % strategy profile in the PageRank model
\newcommand{\PRactions}{\mathcal{\debug{B}}} % strategy in the PageRank model
\newcommand{\PRstr}{\debug \nu}
\newcommand{\PRvec}{\boldsymbol{\PRstr}} % strategy vector in the PageRank model
\newcommand{\PRtrans}{\debug{Q}} % strategy in the PageRank model
\newcommand{\rc}{\mathcal{\debug{C}}} %recurrent component
\newcommand{\residual}{\mathcal{\debug{R}}} %residual component
\newcommand{\scc}{\mathcal{\debug{D}}} % strongly connected component
\newcommand{\skeleton}{\widehat{\tree}}
\newcommand{\tc}{\mathcal{\debug{A}}} %transient closure

\newcommand{\weight}{\debug \omega}
\newcommand{\xomega}{\debug x}

\newacro{eNE}[$\varepsilon$-NE]{$\varepsilon$-Nash equilibrium}
\newacroplural{eNE}[eNE]{$\varepsilon$-Nash equilibria}
\newacro{NE}{Nash equilibrium}
\newacroplural{NE}[NE]{Nash equilibria}
\newacro{PNE}{pure Nash equilibrium}
\newacroplural{PNE}[PNE]{Nash equilibria}
\newacro{PFNE}{prior-free Nash equilibrium}
\newacroplural{PFNE}[PFNE]{weakly prior-free Nash equilibria}
\newacro{WPFNE}{weakly prior-free Nash equilibrium}
\newacroplural{WPFNE}[WPFNE]{weakly prior-free Nash equilibria}
\newacro{WE}{Wardrop equilibrium}
\newacroplural{WE}[WE]{Wardrop equilibria}
\newacro{SO}{socially optimum}

\newacro{KKT}{Karush\textendash Kuhn\textendash Tucker}
\newacro{OD}[O/D]{origin-destination}
\newacro{PoA}{price of anarchy}
\newacro{PoS}{price of stability}
\newacro{PoCS}{price of correlated stability}
\newacro{BPR}{bureau of public roads}
\newacro{FIP}{finite improvement property}
\newacro{eFIP}[$\varepsilon$-FIP]{$\varepsilon$-finite improvement property}
\newacro{BPG}{buck-passing game}
\newacro{SON}{self-organizing network}
\newacro{BHG}{buck-holding game}
\newacro{DBPG}{deterministic buck-passing game}
\newacro{DBHG}{deterministic buck-holding game}
\newacro{SBPG}{stochastic buck-passing game}
\newacro{CBPG}{constrained buck-passing game}
\newacro{CBHG}{constrained buck-holding game}
\newacro{MDBPG}{mixed extension of the deterministic buck-passing game}
\newacro{RWBG}{Random Walk Biasing Game}
\newacro{PR}{PageRank Game}

\begin{document}
%%%%%%%%%%%%%%%%

%\setcounter{page}{1}
%\DOI{10.1287/xxxx.0000.0000}%

 \RUNAUTHOR{Cominetti, Quattropani, and Scarsini}
%\RUNAUTHOR{Jones et al.} % for four or more authors

 \RUNTITLE{The Buck-Passing Game}
%\RUNTITLE{}

 %Full title. Sample:
 \TITLE{The Buck-Passing Game}
 %Enter the full title:
%\TITLE{}

%should be entered in ONE field, separated by a comma. 
%\EMAIL field can be repeated if more than one author
\ARTICLEAUTHORS{%
\AUTHOR{Roberto Cominetti}
\AFF{Facultad de Ingenier\'ia y Ciencias, Universidad Adolfo Ib\'a\~nez, 7941169 Santiago, Chile, \EMAIL{roberto.cominetti@uai.cl} 
\URL{}}
%\AFF{Universidad Adolfo Ib\'a\~nez, \EMAIL{roberto.cominetti@uai.cl} \URL{}}
\AUTHOR{Matteo Quattropani}
\AFF{Dipartimento di Economia e Finanza, Luiss University,  00197 Roma, Italy, \EMAIL{mquattropani@luiss.it} 
\URL{}}
%\AFF{Universit\`a di Roma Tre, \EMAIL{matteo.quattropani@uniroma3.it} \URL{}}

\AUTHOR{Marco Scarsini}
\AFF{Dipartimento di Economia e Finanza, Luiss University,  00197 Roma, Italy, \EMAIL{marco.scarsini@luiss.it} 
\URL{}}
%\AFF{LUISS, \EMAIL{marco.scarsini@luiss.it} \URL{}}
} % end of the block

\ABSTRACT{%

We consider a game in which players are the vertices of a directed graph.
Initially, Nature chooses one player according to some fixed distribution and gives her a buck, which represents the request to perform a chore.
After completing the task, the player passes the buck to one of her out-neighbors in the graph. 
The procedure is repeated indefinitely and each player's cost is the asymptotic expected frequency of times that she receives the buck. 
We consider a deterministic and a stochastic version of the game depending on how players select the neighbor to pass the buck. 
In both cases we prove the existence of pure equilibria that do not depend on the initial distribution; this is achieved by showing the existence of a generalized ordinal potential. 
We then use the \acl{PoA} and \acl{PoS} to measure fairness of these equilibria.
We also study a buck-holding variant of the game in which players want to maximize the frequency of times they hold the buck, which includes the PageRank game as a special case.

%A finite number of agents keep transferring the responsibility of doing a job (the buck) to their neighbors  in a social network. 
%Since the job never gets done, agents want to see the buck coming back to them as rarely as possible.
%This is a game where players are the vertices of a directed graph and the strategy space of each player is the set of her out-neighbors. 
%Nature assigns  the buck to a random player according to a given initial distribution. 
%Each player's cost is the asymptotic expected frequency of times that she gets the buck. 
%We  prove existence of equilibria that do not depend on the initial distribution.
%We analyze fairness of equilibria, and, finally, we discuss a buck holding variant in which players want to maximize the frequency of times they hold the buck. As an application of the latter we briefly discuss the PageRank game.
}%

%\KEYWORDS{deterministic inventory theory; infinite linear programming duality; 
%  existence of optimal policies; semi-Markov decision process; cyclic schedule}
%\MSCCLASS{Primary: 90B05; secondary: 90C40, 90C90}
%\ORMSCLASS{Primary: Inventory/production: deterministic multi-item;
%  secondary: dynamic programming/optimal control: deterministic 
%  semi-Markov; programming: infinite dimensional}
%\HISTORY{Received November 20, 2003; revised March 8, 2004, and March 26, 2004.}

\MSCCLASS{Primary: 91A43; secondary: 91A06, 60J10}
\ORMSCLASS{Primary: Games/group decisions: noncooperative; secondary: probability:
Markov processes}
\HISTORY{}
\maketitle
%%%%%%%%%%%%%%%%%%%%%%%%%%%%%%%%%%%%%%%%%%%%%%%%%%%%%%%%%%%%%%%%%%%%%%
% Samples of sectioning (and labeling) in MOOR.
% NOTE: (1) all section levels end with a period,
%       (2) capitalization is as shown (sentence style, not title style).
%
%\section{Introduction.}\label{intro} %%1.
%\subsection{Duality and the classical EOQ problem.}\label{class-EOQ} %% 1.1.
%\subsection{Outline.}\label{outline1} %% 1.2.
%\subsubsection{Cyclic schedules for the general deterministic SMDP.}
%  \label{cyclic-schedules} %% 1.2.1
%\section{Problem description.}\label{problemdescription} %% 2.
% Text of your paper here
% Appendix here
% Options are (1) APPENDIX (with or without general title) or 
%             (2) APPENDICES (if it has more than one unrelated sections)
% Outcomment the appropriate case if necessary
%
% \begin{APPENDIX}{<Title of the Appendix>}
% \end{APPENDIX}
%
%   or 
%
% \begin{APPENDICES}
% \section{<Title of Section A>}
% \section{<Title of Section B>}
% etc
% \end{APPENDICES}

%
% Section ------------------------------------------
%

\section{Introduction.}
\label{se:intro}

An important chapter of social network analysis concerns reputation systems. 
In the World Wide Web, a hyperlink from one webpage to another is seen as a positive rating of the page that receives the link. 
One of the most successful systems to rank webpages in terms of their reputation is the PageRank algorithm of \citet{BrinPage:1998}, where the rank of each page is determined by the stationary measure of a Markov chain 
whose transition matrix depends on the hyperlinks of the various pages and on a small parameter that represents the probability of jumping to a random page that is not out-linked from the current page.  
The original idea of PageRank was, in the words of its authors, to give a webpage {\em ``an objective measure of its citation importance that corresponds well with people's subjective idea of importance''}.
As some authors point out \citep[see e.g.,][]{AvrLit:SM2006,deKNinvan:LAA2008}, the system is manipulable and reputation of a webpage can be enhanced by strategically choosing its out-links.
The paper by \citet{HopShe:mimeo2008} models this as a game where webpages are players whose strategies are the webpages to which they link, and whose goal is to maximize their reputation ranking.

The PageRank game is a particular instance of a larger class of games where players are vertices of a directed graph, the strategies of each player are its out-neighbors, and the payoff of each player is given by the stationary distribution of a Markov chain induced by the strategy profile. 
Each webpage is usually subjected to some constraints in terms of which other webpages it can link to. 
For instance, a university may not allow its faculty to link their webpages to commercial ones. 
This gives rise to a network of feasible links between webpages.
Each webpage then chooses out-neighbors in this feasible network, subject to other possible constraints such as a maximum number of outlinks.
%\RC{{\em Not sure about what follows}: Would it be pertinent to mention that some links may be more {\em visible} than others (reflecting a higher intensity of
% rating), and therefore might be more likely
%to be followed in a random walk? This is not how PageRank works, but it could serve as a motivation for the more general stochastic framework considered afterwards.}

The goal of this paper is to examine two classes of games with the above features. 
One class, which has the PageRank as its representative, consists of games where the goal of player $\play$ is to maximize the value $\stationary(\play)$, where $\stationary$ is the stationary measure induced by the players' strategy profile.
In the other class of games player $\play$ wants to minimize $\stationary(\play)$. 
An interpretation of the latter is in terms of division of chores among a number of agents, where the division is decentralized and is subject to some underlying network structure. 
Namely, a population of agents is connected by a directed social network. 
Every day one of the agents is in charge of doing a chore, which is onerous to her, but beneficial to the society.
One agent is chosen at random at the beginning of the game and is given a chore to perform.
The agent who is assigned the chore, after doing the job, passes it to one out-neighbor of her choice. 
The goal of each agent is to perform the chore as rarely as possible. 
In a picturesque language, each agent tries to pass the buck to one of her neighbors with the intent of seeing it coming back as seldom as possible.
Here too, some constraints are likely to exist.
Passing the buck could be restricted by geographical proximity, social norms, custom, hierarchy, etc., which determine a network of feasible connections.
For instance the network could be formed by groups of people (families, villages) where only some members of the group have a connection with other groups.
Alternatively, the population may be segregated by gender and the only contacts across  genders may happen within the same family.

%
% Subsection ------------------------------------------
%
\subsection{Our contribution.}

We first deal with the case where players minimize their costs (buck-passing game). 
We study a game on a finite directed graph, where agents are vertices of the graph and directed edges represent the possibility for an agent to transfer the buck to another agent.
We look at situations where the first agent to hold the buck is  drawn at random by Nature.
We first consider the deterministic case in which agents designate 
the neighbor to whom they will transfer the buck, once and for all at the beginning of the game. 
Their goal is to see the buck coming back to them as rarely as possible, and the cost for an agent is the expected asymptotic frequency of times that she gets the buck. 
Although the costs depend on the asymptotic behavior of a process over time, our game is static, since players choose their strategy at the beginning of the game and play the same action whenever their turn comes, without any updating based on the history of the game. 

To establish the existence of pure \aclp{NE}, we prove that the game has a generalized ordinal potential. A classical result by \citet{MonSha:GEB1996} then guarantees that its minimizers are pure \aclp{NE}. 
Moreover, we show there always exist equilibria that are prior-free, i.e., do not depend on the initial distribution according to which Nature makes its draw. 
In general the game may have multiple \aclp{NE} and some of them might be prior-sensitive.

We then look at a stochastic version of the game in which the agents choose the probability with which the buck is passed to each of their neighbors. 
This gives rise to a Markov chain. When this chain is irreducible, its unique stationary distribution is precisely the cost vector of the game, and does not depend on Nature's initial distribution. 
In general, we prove the existence of a generalized ordinal potential and we use it to establish the existence of prior-free \aclp{NE}.
To this end we exploit the Markov chain tree theorem to derive an explicit formula for the potential function.
We also provide an alternative characterization for the potential, showing that it can be written  as the expected length of unicycles in the graph with respect to a suitable probability measure. 
All of this establishes a new mathematically intriguing bridge between the \acl{BPG} and Markov chains.

We next investigate fairness of the equilibria in \aclp{BPG}, that is, we study how unevenly the total cost is spread across players in equilibrium in comparison to what could be achieved by a benevolent planner who wants to minimize disparity of treatment. 
In the spirit of \citet{Raw:HUP2009}, we define the social cost function of a strategy profile as the highest cost across all players.
Then we use the \acl{PoA}\acused{PoA} and the \acl{PoS}\acused{PoS} to measure fairness. 
Typically these quantities are used to measure efficiency of the worst and the best equilibrium, respectively: the social cost is usually taken to be the sum of the costs of all the players. 
Since our \acl{BPG} is a constant-sum game, efficiency is not an issue, but, using a Rawlsian social cost function the  \acl{PoA}\acused{PoA} and the \acl{PoS}\acused{PoS} can be used as a  measure of fairness. 

In the last section we turn to \aclp{BHG}, which have the same structure as \aclp{BPG}, except that the cost becomes a payoff and the goal of each player is to see the buck coming back as often as possible. 
Mathematically, the analysis of \aclp{BHG} follows the same line as that of \aclp{BPG}, showing that  the \acl{BHG} admits a generalized ordinal potential function that can be obtained by reversing the sign of the \acl{BPG} potential function.
We close our paper with a brief discussion of the PageRank game as a special case of a \acl{BHG}.

%
% Subsection ------------------------------------------
%
\subsection{Related literature.}

Various classes of games on networks have been considered in the literature. 
In some of them, which go under the name ``network games,'' the payoff of a player depends only on her own strategy and on the strategies of her neighbors \citep[see, e.g.,][among many]{KeaLitSin:PCUAI2001,GalGolJacVegYar:RES2010,ParOzd:GEB2019}.
This is not the case in our model, where the payoff of each player depends on the strategies of all the other players.

As already mentioned, our buck-passing game can be seen as a decentralized strategic procedure to split chores under network constraints. 
In this perspective, it is somehow linked to the literature on fair division. 
This literature is huge and \citet{Mou:ARE2019} provides a nice recent survey.
Measuring fairness in the framework of online fair division is studied in 
\citet{BogMouSam:arXiv2020}, where a measure called price of fairness is proposed.  
An interesting connection between potential games and fair allocation has been studied by 
\citet{GopMarWie:MOR2014}, who prove that existence of a pure equilibrium in distribution rules having a fixed local welfare function requires the game to be potential.
Our fairness criterion is inspired by the work of   \citet{Raw:HUP2009}.
We adapted to this criterion the typical measures of inefficiency, i.e.,   the \acl{PoA}\acused{PoA}  
\citep[][]{KouPap:STACS1999,KouPap:CSR2009,Pap:PACM2001}, and the \acl{PoS}\acused{PoS}
\citep[][]{SchSti:P14SIAM2003,AnsDasKleTarWexRou:SIAMJC2008}.
In most of the literature the social cost is the sum of the costs of all the players. 
Other social cost functions were considered, for instance, in 
\citet{KouPap:CSR2009,KouPap:STACS1999},
\citet{Vet:FOCS2002}, \citet{MavMonPap:Springer2008}, \citet{FouSca:MOR2019}.

The seminal paper by \citet{BrinPage:1998} introduced the PageRank dynamics as a tool to rank webpages, the ancestor of the algorithm used nowadays by Google to produce  an ordered list of pages as output of a query.
In the past decade, PageRank has been intensively studied in both the theoretical and applied literature,  \citep[see, e.g.,][]{
	JehWid:P12WWW2003,
	AndChuLan:IM2008,
	AvrLitSon:IM2008,
CheLitOlv:RSA2017,CapQua:arXiv2019,
	LeeOlv:SPA2020,GarHofLit:AAP2020}. 
The idea of looking at the PageRank in a strategic setting goes back to \citet{AvrLit:INRIA2004,AvrLit:SM2006,deKNinvan:LAA2008}. 
In \citep{AvrLit:SM2006} the effect of adding a new outgoing link to a webpage in PageRank is studied and an optimal strategy is proposed.
A similar problem is studied in \citep{deKNinvan:LAA2008} in the case where a webmaster controls a subset of webpages.
The first game theoretical formulation of the PageRank can be found in \citet{HopShe:mimeo2008}
who prove the existence of Nash equilibria and study their features \citep[see also][]{HopShe:IM2008}.
A similar model was studied shortly after by \citet{CheTenWanZho:FA2009}.
\citet{AviIwaPak:DAM2014} consider versions of the PageRank game on undirected networks, where players cannot unilaterally create links, but they can delete existing links.
\citet{KouMarPapRigSid:SAGT2015} deal with a game of the PageRank type where players choose both their outgoing links and their weights. 
Using quasi-concavity arguments they prove existence of pure Nash equilibria.
A variant studied by \citet{CasCatComFag:arXiv2020} consider a game where players can choose how to direct their $m$ links in order to maximize their Bonacich centrality.

%\MS{Cite
%\citet{TraNesDoo:SI2010}
%\citet{LeFCesGenVit:PMLR2017}
%\citet{TanChaAggLiu:ACMCS2016}
%}
The main tools that we use to prove the existence of pure \aclp{NE} are the existence of a generalized ordinal potential and the \acl{FIP}. 
The relationship between these concepts was studied by \citet{MonSha:GEB1996}.
We also rely on classical results for Markov chains, for which we refer to the books of \citet{Nor:CUP1998,AldFil:mono2002,LevPer:AMS2017}. 
In particular we exploit the celebrated Markov chain tree theorem, attributed to \citet{VenFre:UMN1970} \citep[see also][]{LeiRiv:IEEETIT1986,AnaTso:SPL1989}, which relates the stationary measure of the chain to the abundance of spanning trees in the underlying graph. 
Our model is also related to a research stream that connects Markov chains with the classical Hamiltonian cycle problem, 
\citep[see, e.g.,][]{FilKra:MOR1994,BorEjoFil:RSA2009,BorEjoFil:RSA2004,BorEjoFilNgu:Springer2012,EjoFilNgu:MOR2004,LitEjo:MOR2009,EjoLitNguTay:JAP2011,EjoFilMurNgu:SJDM2008}. 
We defer the discussion of these connections to \cref{suse:Hamiltonian}.

\subsection{Organization of the paper.}
\label{suse:organization}

The paper is organized as follows.
\cref{se:graphs} introduces the notation.
\cref{se:deterministic} analyzes the deterministic version of the buck-passing game. 
\cref{se:stochastic} introduces the stochastic model.
\cref{se:fairness} studies  fairness of the equilibria.
\cref{se:Markov-chains} introduces the probabilistic tools which are needed to show the existence of a generalized ordinal potential function.
\cref{se:gop} deals with the potential nature of the game.
\cref{se:BHG} examines the class of buck-holding games and explores their connection with the PageRank game.

%	
%Section ------------------------------------------
%

\section{Graph terminology and notations.}
\label{se:graphs}

Throughout the paper we consider a directed graph $\graph=(\vertices,\edges)$, which represents a social network, with $\vertices=\{1,\dots,\nPlayers\}$ the set of vertices and  $\edges\subseteq\vertices\times\vertices$ the set of edges. 
We further assume that $\graph$ has no loops.  
The following standard terminology will be used hereafter:

\begin{enumerate}[(a)]
	\item\label{it:out-neighbors} The set of \emph{out-neighbors} of vertex $\play$ is denoted by $\neighborsout_{\play}=\{\playalt:(\play,\playalt)\in\edges\}$. 
	Its cardinality $|\neighborsout_{\play}|$ is called the \emph{out-degree} of the vertex.
	
	\item\label{it:path}
	A \emph{path} is a sequence of edges $\edge_{1},\dots,\edge_{\run}$ where, for all $\play\in\{1,\dots,\run-1\}$, the head of $\edge_{\play}$ coincides with the tail of $\edge_{\play+1}$.
	
	\item\label{it:strongly-connected} 
	The graph $\graph$  is \emph{strongly connected} if for every  $\play,\playalt\in\vertices$ there exists a path from $\play$ to $\playalt$.
	
	\item\label{it:subgraph}
	A \emph{subgraph} is a graph $\graph(\vertices',\edges')$ with $\vertices'\subseteq \vertices$ and $\edges'\subseteq\edges$.
	
	\item\label{it:spanning}
	A \emph{spanning subgraph} is a subgraph $\graph(\vertices',\edges')$ where $\vertices'=\vertices$.
	
	\item\label{i:simle-cycle}
	A \emph{cycle} is a strongly connected graph where each vertex has out-degree $1$.

	\item\label{it:unicycle}
	A \emph{unicycle} is a graph where each vertex has out-degree 1 and which contains exactly one cycle. 
	
	\item\label{it:rooted-tree}
	An \emph{$\play$-rooted tree} is a graph that contains no cycles and such that $\play$ has out-degree $0$ and the other vertices  have out-degree $1$.  
\end{enumerate}

\begin{figure}[h]
	\FIGURE
	{\includegraphics[width=4cm]{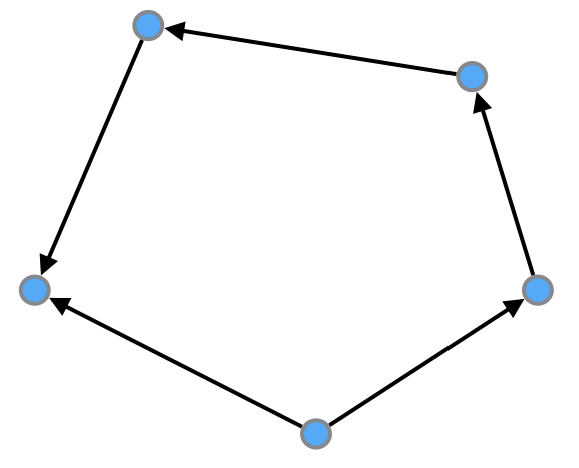}\qquad
		\includegraphics[width=4cm]{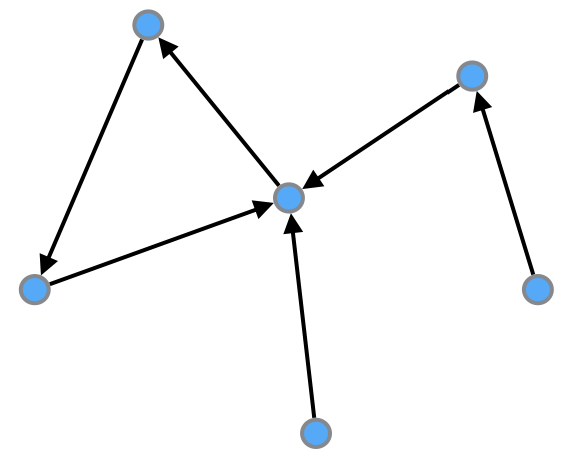}\qquad
		\includegraphics[width=4cm]{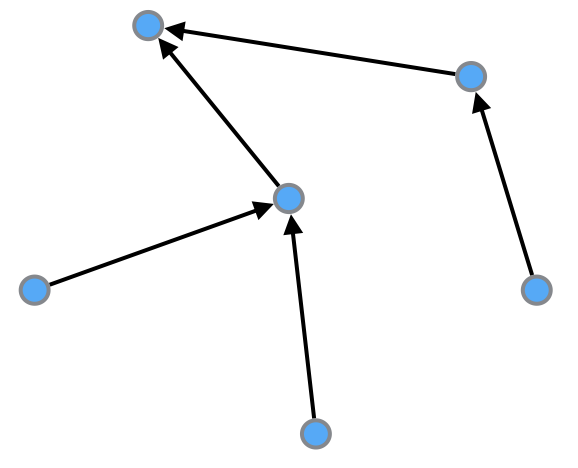}
	}
	{
		Left: a cycle. Middle: a unicycle. Right: a rooted tree. Note that removing  any edge from the cycle of a unicycle we get a rooted tree.
		\label{fig:example} }
	{}
\end{figure}

\cref{fig:example} shows examples of a cycle, a unicycle, and a rooted tree. 
As immediate from their definition, not every unicycle is a cycle, but it always contains one.

%
% Section ------------------------------------------
%
\section{The deterministic buck-passing game.}
\label{se:deterministic}

%We start by introducing a deterministic version of the \acl{BPG}.
%
%%
%% Subsection ------------------------------------------
%%
%
%\subsection{The game}\label{suse:DBPG}

Let $\graph=\parens{\vertices,\edges}$ be a directed graph and $\initial=\parens{\initial_{\play}}_{\play\in\vertices}$ a probability distribution over the vertices.
We consider a finite game $\game(\graph,\initial,\actions)$ where each vertex $\play\in\vertices$ is a player with strategy set $\actions_{\play}=\neighborsout_{\play}$ (assumed nonempty with $\play\not\in \actions_{\play}$), and $\actions=\times_{\play\in\vertices}\actions_{\play}$ is the set of strategy profiles. 

Once each player has chosen an out-neighbor $\act_{\play}\in\actions_{\play}$, the cost for a player is the asymptotic frequency of times she has the buck, determined by the following process. 
At time $\per=0$ a buck is given to a vertex $\play_{0}\in\vertices$ drawn at random by Nature according to the initial distribution $\initial$. 
At time $\per=1$, the selected player $\play_{0}$ passes the buck to her designated neighbor $\play_{1}=\act_{\play_{0}}$, who in turn will pass it at time $\per=2$ to her chosen neighbor, and so on. Define the random variables
\begin{equation}
\label{eq:buck}
\buck_{\play,\per}(\actprof) =
\begin{cases}
1 & \text{if at time $\per$ player $\play$ has the buck},\\
0 & \text{otherwise}.
\end{cases}
\end{equation}
For a fixed profile $\actprof$, the value of  $\buck_{\play,\per}$ depends only on the initial draw, with
\begin{equation}\label{eq:mu-Theta}
\prob(\buck_{\play,0}(\actprof)=1)=\initial_{\play}.
\end{equation}
After this initial draw, the buck is passed among the players and eventually it will start cycling, so that we can define the cost function $\cost_{\play}:\actions\to\reals$ for player $\play$ as 
\begin{equation}
\label{eq:SBPG-cost}
\cost_{\play}(\actprof)=\expect\bracks*{\lim_{\horizon\to\infty}\frac{1}{\horizon}\sum_{\per=1}^{\horizon}\buck_{\play,\per}(\actprof)},
\end{equation}
where the expectation is taken with respect to the initial measure $\initial$.

This game is denoted $\game(\graph,\initial,\actions)$ and is called a \acfi{DBPG}\acused{DBPG}. 
The corresponding set of \aclp{NE} is denoted by $\NE(\actions)$.
We assume that the graph $\graph$, the initial measure $\initial$, and the buck-passing dynamics are common knowledge.
We stress that, despite the fact that the costs are defined through a dynamic process, the game is actually static, with strategies fixed once and for all at the beginning of the game. 
Note also that the costs of all players add up to $1$, so that this game is equivalent to a zero-sum game. 
To analyze it, it is convenient to have a more manageable expression for the costs in \eqref{eq:SBPG-cost},
for which we introduce some additional notation.

\begin{definition}
	\label{de:InducedGraph}
	We use the notation  $\graph_{\actprof}=\parens{\vertices,\edges_{\actprof}}$ for the subgraph induced by the strategy profile $\actprof$ with edge set $\edges_{\actprof}=\braces{(\play,\act_{\play}):\play\in\vertices}$. 
	Each vertex has out-degree $1$ so that $\graph_{\actprof}$ is a union of a finite number $\numrc(\actprof)$ of disjoint unicycles
	(see \cref{fig:notsc-unicycle}).
	For $\countrc=1,\ldots,\numrc(\actprof)$, the vertex set of  unicycle $\countrc$ is denoted by $\tc_{\actprof}^{\countrc}$, so that $\vertices=\tc_{\actprof}^{1}\cupdot\cdots\cupdot\tc_{\actprof}^{\numrc(\actprof)}$, and $\rc_{\actprof}^{\countrc}\subseteq \tc_{\actprof}^{\countrc}$ denotes the set of vertices in the corresponding cycle. 
\end{definition}

\begin{figure}[h]
	\FIGURE
	{\includegraphics[width=7cm]{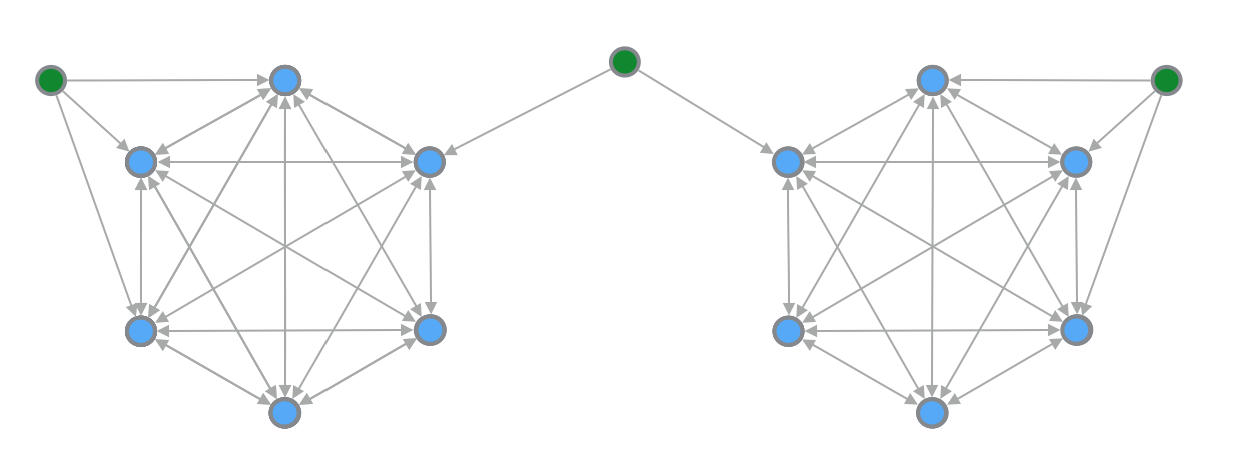}
	\includegraphics[width=7cm]{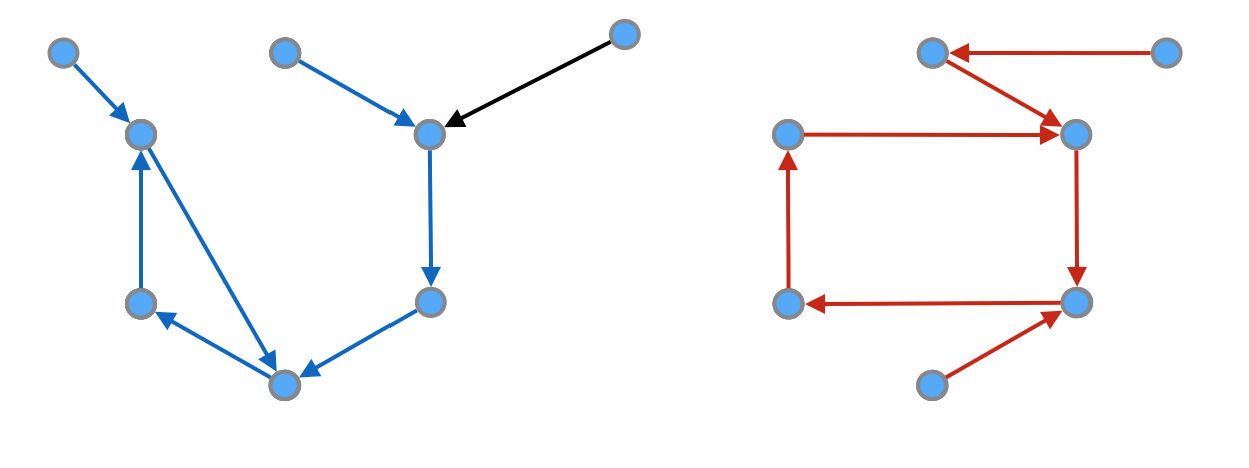}}
	{
		The  graph $\graph$ on the left has 2 strongly connected components and 3 transient vertices.
		The induced graph $\graph_{\actprof}$ on the right has two disjoint unicycles with cycles of length 3 and 4.
	\label{fig:notsc-unicycle} }
{}
\end{figure}

If the buck is assigned initially to a vertex in $\tc_{\actprof}^{\countrc}$, then, after finitely many steps, it will reach the cycle $\rc_{\actprof}^{\countrc}$ and turn around forever. 
In the long run each player $\play\in\rc_{\actprof}^{\countrc}$ gets the buck a fraction $1/\abs{\rc_{\actprof}^{\countrc}}$ of times, whereas the remaining players are free-riders with a cost of $0$. 
Now, the probability that the buck is  assigned initially to a vertex in $\tc_{\actprof}^{\countrc}$ is 
\begin{equation}
\label{eq:initial-ell}
\initial_{\actprof}^{\countrc} \coloneqq \sum_{\playalt\in\tc_{\actprof}^{\countrc}}\initial_{\playalt}.
\end{equation}
Hence, defining $\countrc(\play)$ the label of the unicycle that contains player $\play$ and setting 
\begin{equation}\label{eq:ind-unicycle}
\indunic_{\play}(\actprof) \coloneqq
\begin{cases}
1 & \text{if } \play\in\rc_{\actprof}^{\countrc(\play)},\\
0 & \text{otherwise},
\end{cases}
\end{equation}
the expected cost in \eqref{eq:SBPG-cost} can be written as
\begin{equation}\label{eq:DBPG-cost}
\cost_{\play}(\actprof)=\dfrac{\initial_{\actprof}^{\countrc(\play)}}{\abs{\rc_{\actprof}^{\countrc(\play)}}}\indunic_{\play}(\actprof).
\end{equation}

%
% Section ------------------------------------------
%
\subsection{Ordinal potentials and existence of prior-free equilibria.}
\label{suse:prior-free}

As every finite game, the buck-passing game admits equilibria in mixed strategies. 
However, our main interest here is the existence of \aclp{NE} in \emph{pure strategies}, which  for general games are not guaranteed to exist. 
Unless otherwise stated, we always refer to equilibria in pure strategies. 
We now recall the concepts of profitable deviations and  equilibria.
\begin{definition} \label{de:nash-enash} 
	Consider a cost game.
	\begin{enumerate}[(a)]
		\item
		Given a strategy profile $\actprof\in\actions$, a \emph{unilateral deviation} for player $\play$ is a strategy $\actprofalt\in\actions$ that 
		differs from $\actprof$ only in its $\play$-th coordinate. It is  a \emph{profitable deviation} if in addition
		$\cost_{\play}(\actprofalt)<\cost_{\play}(\actprof)$, in which case the difference $\cost_{\play}(\actprof)-\cost_{\play}(\actprofalt)$ 
		is called the \emph{improvement} of player $\play$.
		
		\item
		A  strategy profile $\actprof\in\actions$ is a \acfi{NE}\acused{NE} if no player has a profitable deviation.
		Similarly, it is an \acfi{eNE}\acused{eNE} if no player has a profitable deviation with an improvement larger than $\varepsilon$.
	\end{enumerate}
\end{definition}

In principle, a \acl{BPG}  $\game(\graph,\initial,\actions)$ may have multiple equilibria and they  may depend on the initial measure $\initial$. 
Of special interest are the so-called \emph{prior-free equilibria}, i.e., equilibria that are invariant with respect to the initial measure $\initial$.

\begin{definition}\label{de:prior-free}
A strategy profile $\actprof\in\actions$ in the \acl{BPG} is called a \acfi{PFNE}\acused{PFNE} if  $\actprof$ is a Nash equilibrium of $\game(\graph,\initial,\actions)$
for every initial distribution $\initial$.
\end{definition}

\begin{example}\label{ex:multiple-eq}
	Consider  the graph $\graph$  on the top of \cref{fig:prior-free-ce}, and suppose that  $\initial$ is a degenerate measure that puts all the mass on the vertex labeled by $v$.
		Then the strategy profile shown on the bottom left picture in \cref{fig:prior-free-ce} is a \ac{NE} which is not prior-free.  
		Indeed, if we take a different measure $\initial'$ which puts some positive mass on one of the vertices $i,j,k$, then $k$ has a profitable deviation, which gives rise to the \ac{PFNE} shown on the middle picture. Notice that a \emph{more fair} \ac{PFNE} is given by a cycle passing by every vertex, which is represented on the right of \cref{fig:prior-free-ce}. 	
	\begin{figure}[h]
		\centering
		\includegraphics[width=4.5cm]{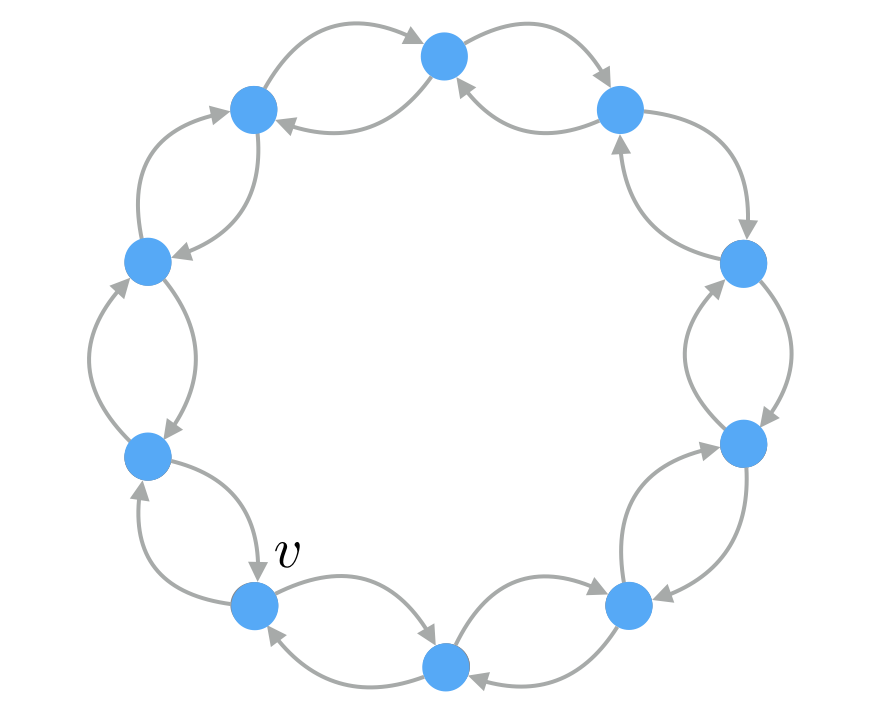}\\
		\includegraphics[width=4.5cm]{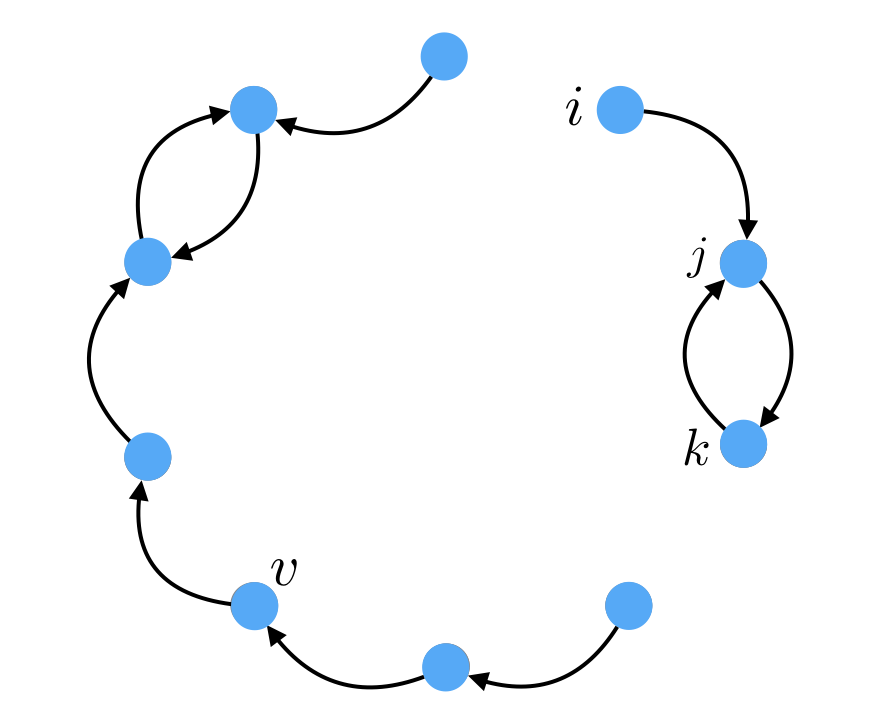}
		\includegraphics[width=4.5cm]{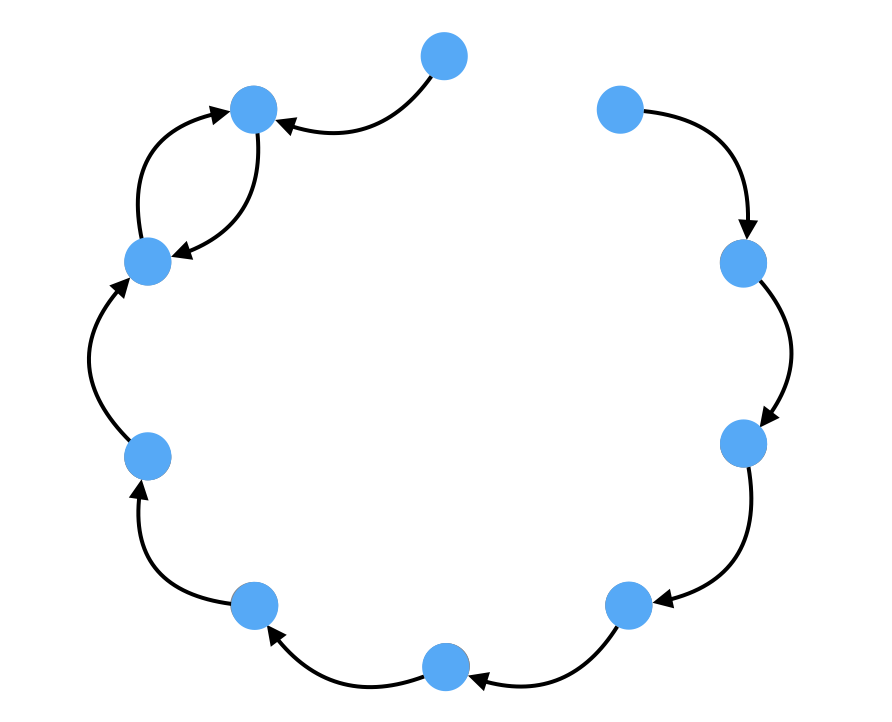}
		\includegraphics[width=4.5cm]{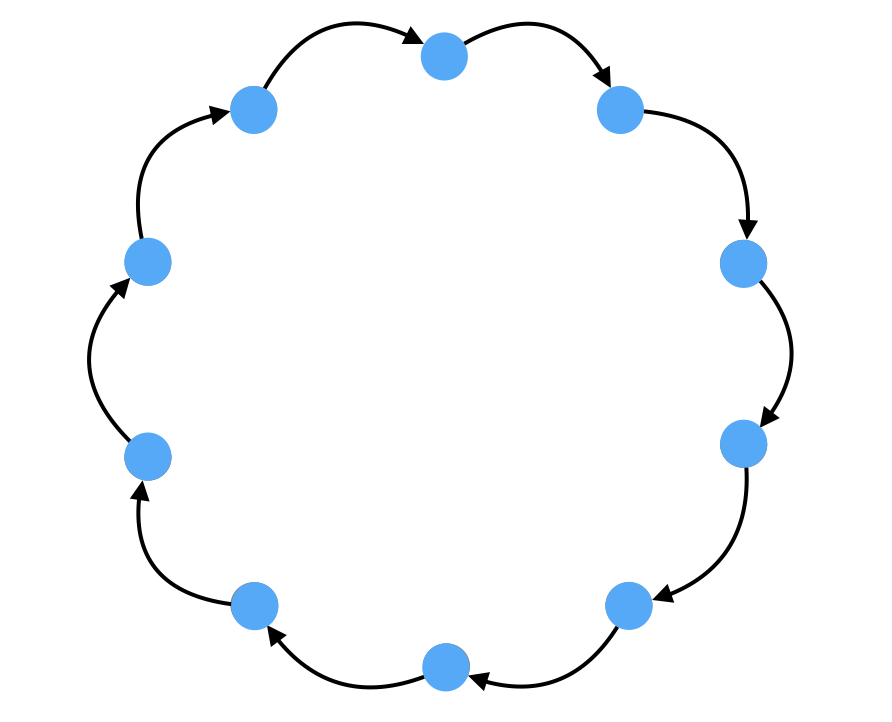}
		\caption{Top: The original graph. In this game each player can pass the buck only clockwise or counterclockwise. Left: a \ac{NE} that is not prior-free. Middle: a  \ac{PFNE}. Right: another  \ac{PFNE}.}
		\label{fig:prior-free-ce}
		
	\end{figure}
\end{example}

Our main result for the deterministic buck-passing game is the existence of   prior-free equilibria.
This will be proved by showing the existence of a generalized ordinal potential that does not depend on the initial measure $\initial$.
Potential games were introduced in the seminal paper by \citet{Ros:IJGT1973} and later studied extensively  by \citet{MonSha:GEB1996}. 
A recent account can be found in \citet{LaCheSoo2016}.
We recall these notions for a general cost game $\parens{\vertices,\actions,\costs}$, where $\vertices$ is a finite set of players, $\actions=\times_{\play\in\vertices}\actions_{\play}$ is the set of strategy profiles  with $\actions_{\play}$ the (possibly uncountable) set of strategies for player $\play\in\vertices$, and $\costs_{\play}:\actions\to\reals$ is the cost of player $\play\in\vertices$.

\begin{definition}
	\label{de:ordinal-potential} 
	A function $\potential:\actions\to\reals$ is called a \emph{potential} for $\parens{\vertices,\actions,\costs}$
	if for each strategy profile $\actprof$ and any unilateral deviation $\actprofalt=\parens{\actalt_{\play}, \actprof_{-\play}}$ 
	by a player $\play$, we have 
	\begin{equation}\label{eq:potential}
	\cost_{\play}(\actprofalt)-\cost_{\play}(\actprof)=\potential(\actprofalt)-\potential(\actprof).
	\end{equation}
	Similarly, $\potential$  is said to be an  \emph{ordinal potential} if
	\begin{equation}\label{eq:ord_potential}
	\cost_{\play}(\actprofalt)<\cost_{\play}(\actprof) \Leftrightarrow 
	\potential(\actprofalt)<\potential(\actprof)
	\end{equation}
	and  is called a \emph{generalized ordinal potential} if		
	\begin{equation}\label{eq:gen_ord_potential}
	\cost_{\play}(\actprofalt)<\cost_{\play}(\actprof) \Rightarrow 
	\potential(\actprofalt)<\potential(\actprof).
	\end{equation}
\end{definition}

Clearly the notion of generalized ordinal potential is the weakest. It is also easy to see that a component-wise minimizer of a generalized ordinal potential is a Nash equilibrium. 
If $\potential$ is in fact an ordinal potential, these minimizers coincide with the Nash equilibria. 
For later reference we record the following direct consequence.

\begin{proposition}\label{pr:GOP-pure_NE}
	Every generalized ordinal potential game, which is either finite or, more generally, which has compact strategy sets and a lower semi-continuous potential, admits a Nash equilibrium in  pure strategies. 
\end{proposition}

\begin{definition}\label{de:FIP} \
	\begin{enumerate}[(a)]
		\item
		An \emph{improvement path}  is a sequence of strategy profiles $\actprof_{0}, \actprof_{1}, \dots$ such that each $\actprof_{\run+1}$ is a profitable deviation of $\actprof_{\run}$ for some player $\play_{\run}$. 
		It is called an \emph{$\varepsilon$-improvement path} if the  improvement at each stage is at least $\varepsilon$.
		
		\item
		A game has the \acfi{FIP}\acused{FIP} if every improvement path is finite.
		Similarly, it has the \acfi{eFIP}\acused{eFIP} if every $\varepsilon$-improvement path is finite.
	\end{enumerate}
	
\end{definition}

\citet{MonSha:GEB1996} showed that a finite game has a generalized ordinal potential if and only if it satisfies the \ac{FIP}.
For a similar characterization of ordinal potentials see \citet{VooNor:GEB1997}.

\begin{remark}\label{re:FIP}
	The  \ac{FIP} can also be described by means of an auxiliary graph---which we call the \emph{meta-graph} of the game to avoid confusion with the other graphs mentioned in the paper---where a vertex represents a strategy profile and a directed edge  $(\actprof,\actprofalt)$ exists iff $\actprofalt$ is a profitable deviation from $\actprof$ for some player $\play$. 
	Then, the \ac{FIP} is equivalent to the acyclicity of the meta-graph. 
	A similar construction is used, for instance, in \citet{CanMenOzdPar:MOR2011}. 
	Note that \ac{NE} are precisely the sinks of the meta-graph, i.e., 
	meta-vertices with out-degree equal to zero. 
	Since every path on an acyclic directed graph is finite and ends in a sink, this shows again that a finite game satisfying the \ac{FIP} admits a \ac{NE}.
\end{remark}

With these preliminaries, we proceed to establish our main results for the deterministic buck-passing game. 
We will prove the existence of a generalized ordinal potential, from which we will deduce not only that its minima are Nash equilibria in pure strategies but also that they are prior-free. 
Furthermore, we establish a bound for the maximum length of an improvement path, which implies that any best response dynamics must reach a prior-free equilibrium in quadratic time.

Since a cost that is already at zero cannot be decreased, it follows from \eqref{eq:DBPG-cost}  that a profitable deviation may only concern a player $\play$ who is on a cycle $\rc_{\actprof}^{\countrcalt}$ with $\initial_{\actprof}^{\countrcalt}>0$.
Such a player has two options to reduce her cost: increase the length of the cycle she is currently in, or break the cycle by sending the buck to a different unicycle so that her cost drops to zero. More precisely these options are:
\begin{enumerate}[({\textsc{d}}$_1$)]
	\item
	\label{it:deviation-D1}
	shift to $\actalt_{\play}\in \tc_{\actprof}^{\countrcalt}\setminus \rc_{\actprof}^{\countrcalt}$
	so that the cycle becomes longer,  
	
	\item
	\label{it:deviation-D2}
	shift to $\actalt_{\play}\not\in \tc_{\actprof}^{\countrcalt}$ so that the cycle $\rc_{\actprof}^{\countrcalt}$  disappears.
\end{enumerate}
These deviations affect only the cycle of the deviating player, either by increasing its length, or 
by breaking it and merging $\tc_{\actprof}^{\countrcalt}$ altogether into some other unicycle. 
All the other cycles remain unaltered, even if their unicycles may change.
Our first result identifies a potential function based on the length of these cycles.

\begin{theorem}\label{th:GOP}
	The \acl{DBPG} $\game(\graph,\initial,\actions)$ is a generalized ordinal potential game with generalized ordinal potential function
	\begin{equation}\label{eq:potentialdef}
	\potential(\actprof)\coloneqq \sum_{\countrc=1}^{\numrc(\actprof)}\parens*{\nPlayers-\abs{\rc_{\actprof}^{\countrc}}},
	\end{equation}
	where $\rc_{\actprof}^{\countrc}$ are the cycles in the induced subgraph
	$\graph_{\actprof}=\parens{\vertices,\edges_{\actprof}}$ (see \cref{de:InducedGraph}).
	If the initial measure $\initial$ has full support, then $\potential$ is an ordinal potential.
\end{theorem}

\proof{Proof.} Consider a strategy profile $\actprof\in\actions$ and let $\countrcalt=\countrc(\play)$ be the unicycle containing player $\play$ in $\graph_{\actprof}$. Consider also a deviation by player $\play$ from $\actprof$ to $\actprofalt=(\actalt_{\play},\actprof_{-\play})$.
	
	If the deviation is profitable, that is,  $\cost_{\play}(\actprofalt)<\cost_{\play}(\actprof)$ then, as noted above, there are only two possible cases.
	After a deviation \ref{it:deviation-D1} the new graph $\graph_{\actprofalt}$ has the same cycles except for $\rc_{\actprof}^{\countrcalt}$ which becomes longer, so that only this term in the sum in \eqref{eq:potentialdef} is affected and $\potential(\actprofalt)<\potential(\actprof)$. 
	On the other hand, a deviation \ref{it:deviation-D2} removes the cycle $\rc_{\actprof}^{\countrcalt}$ keeping the other cycles unchanged, so that we lose
	one term in \eqref{eq:potentialdef} and again $\potential(\actprofalt)<\potential(\actprof)$.
	This proves that  $\potential$ is a generalized ordinal potential.
	
	We  next show that, when $\initial$ has full support, the reverse implication holds, that is, $\potential(\actprofalt)<\potential(\actprof)$ implies $\cost_{\play}(\actprofalt)<\cost_{\play}(\actprof)$. 
	The inequality $\potential(\actprofalt)<\potential(\actprof)$ 
	conveys a change in the structure of cycles, which can only occur if the deviating player $\play$ is on a cycle $\rc_{\actprof}^{\countrcalt}$. 
	The deviation can only affect this cycle by removing it or changing its length, so again we distinguish these two cases. 
	In the first case a reduction of the potential requires the length of the cycle to increase $\abs{\rc_{\actprofalt}^{\countrcalt}}>\abs{\rc_{\actprof}^{\countrcalt}}$, which imposes $\actalt_{\play}\in \tc_{\actprof}^{\countrcalt}\setminus \rc_{\actprof}^{\countrcalt}$. 
	In this case $\tc_{\actprofalt}^{\countrcalt}= \tc_{\actprof}^{\countrcalt}$ so that $\initial_{\actprofalt}^{\countrcalt}=\initial_{\actprof}^{\countrcalt}>0$
	and $\cost_{\play}(\actprofalt)<\cost_{\play}(\actprof)$ follows from \eqref{eq:DBPG-cost}.
	The second case occurs when $\actalt_{\play}\not\in \tc_{\actprof}^{\countrcalt}$, in which case $\cost_{\play}(\actprofalt)=0<\cost_{\play}(\actprof)$, where the strict inequality follows from \eqref{eq:DBPG-cost} by noting that $\indunic_{\play}(\actprof)=1$ and $\initial_{\actprof}^{\countrcalt}>0$.
\Halmos
\endproof

\begin{theorem}\label{th:prior-free-equilibrium}
	Every \acl{DBPG} has a \ac{PFNE}. 
\end{theorem}
\proof{Proof.} Since for every initial distribution $\initial$ the game $\game(\graph,\initial,\actions)$ is a finite game, the existence of a pure Nash equilibrium follows directly from \cref{th:GOP} and \cref{pr:GOP-pure_NE}. To prove the existence of a prior-free equilibrium it suffices to note that the expression $\potential$ in \eqref{eq:potentialdef} does not depend on the initial distribution $\initial$, so that a global minimizer of $\potential$	is a \ac{PFNE}.
\Halmos
\endproof

We stress that, when $\initial$ is fully supported, $\potential(\actprof)$ provides an ordinal potential and \ac{NE} are exactly its component-wise minimizers. 
Since  $\potential(\actprof)$ does not depend on $\initial$ it follows that in this case all \ac{NE} are prior-free. 
In other words, prior-sensitive equilibria can  appear only when $\initial$ is not fully supported. 
We refer again to \cref{ex:multiple-eq} for an easy example.

\begin{remark}
\label{re:Hamilton}
The following immediate consequence of our analysis will turn useful later. 
Assume that $\graph$ admits a Hamiltonian cycle $\hamil$, that is, a directed cycle passing for every vertex in $\vertices$ exactly once. 
Call $\actprof_{\hamil}\in\actions$ the strategy profile that induces the Hamiltonian cycle. 
Then $\actprof_{\hamil}$ is a \ac{PFNE} for the \ac{DBPG}. 
In fact, no  player has an incentive in deviating, since the deviation will cause a decrease in the length of the cycle. 
We refer the reader to \cref{suse:Hamiltonian} for a more detailed analysis of the connections between the the \acl{BPG} and the existence of Hamiltonian cycles.  
\end{remark}

In view of \cref{th:GOP}, and according to \citet{MonSha:GEB1996},  it follows that every \acl{DBPG} has the \ac{FIP}. 
Furthermore, considering the two types of profitable deviations \ref{it:deviation-D1} and \ref{it:deviation-D2}, we see that, if in a given strategy profile $\actprof\in\actions$ two vertices $\play$ and $\playalt$ are in the same unicycle, the same  holds for every $\actprofalt\in\actions$ that can be reached by following an improvement path. 
To be picturesque, we could say that, once the destinies of two players meet, they are doomed to be entangled forever.
In the language of probability, the evolution along an improvement path is a \emph{coalescence process}. 
This observation leads to an explicit bound on the maximum length of an improvement path in terms of the number of 
players $\nPlayers$, and it implies that a best response dynamics will attain a \ac{PFNE} in quadratic time.

\begin{theorem}\label{th:DBPG-FIP}
	In every \acl{DBPG} $\game(\graph,\initial,\actions)$ the length of an improvement path is at most $\frac{1}{4}\nPlayers^2-1$,
	and this bound is tight.
\end{theorem}

\proof{Proof.}
	Let $\phi_k$ be the maximum length of an improvement path when we start from a strategy profile with 
	$\numrc(\actprof)=k$ unicycles. As noted before, there are two types of profitable deviations, \ref{it:deviation-D1} and \ref{it:deviation-D2},
	in which either the length of a cycle increases or a unicycle merges into another. Since each unicycle contains at least 2 vertices,
	there can be at most $\nPlayers-2k$ deviations of type \ref{it:deviation-D1} before dropping in the next deviation to $k-1$ unicycles.
	This yields the bound
\[
\phi_k\leq (\nPlayers-2k+1)+\phi_{k-1}
\]
	and inductively we get
\[
\phi_k\leq \sum_{i=2}^k(\nPlayers-2i+1)+\phi_{1}=k\nPlayers-k^2+1-\nPlayers+\phi_1.
\]
	By the same argument as above we have $\phi_1\leq \nPlayers-2$ which yields
\[
\phi_k\leq k\nPlayers-k^2-1.
\]
	The maximum of the last expression is attained at $k=\lfloor\frac{\nPlayers}{2}\rfloor$. Ignoring the rounding and maximizing
	for $k\in\reals$ we get $\phi_k\leq \frac{1}{4}\nPlayers^2-1$.
	
	We next show that this bound can be reached.
	Consider a complete graph $\completegraph_{\nPlayers}$ with an even number of vertices and  a uniform initial measure $\initial$. 
	Take an initial strategy profile with 
	exactly $k=\nPlayers/2$ cycles $\rc_{1},\dots,\rc_{k}$ with $2$ players each, and
	consider the following sequence of unilateral profitable deviations:
	\begin{itemize}
		\item[(1)] Break the cycle $\rc_{1}$ by connecting it as a path to the cycle $\rc_{2}$. 
		This counts as a single step, after which the cycle $\rc_{2}$ has a ``tail'' composed by the vertices in $\rc_{1}$. 
		Call $\play$ the vertex in $\rc_{2}$ to which such tail is connected and let $\playalt$ be the predecessor of $\play$ in $\rc_{2}$. 
		Note that $\playalt$ can reduce her cost by connecting to any vertex of the tail:
		the farther the selected player in the tail, the lower the resulting cost for the deviating player $\playalt$. 
		Consider the scenario where $\playalt$ chooses the \emph{worst} profitable deviation, by selecting the closest player in the tail as new designated out-neighbor.
		Repeat this procedure until $\rc_{1}$ and $\rc_{2}$ are merged into a unique cycle $\rc_{1,2}$ with 4 vertices.
		This requires 2 steps. 
		
		\item[(2)] Break the cycle $\rc_{1,2}$ by connecting it as a path to $\rc_{3}$. This requires $1$ step.
		Enlarge the cycle $\rc_{3}$ by collecting one by one the elements of $\rc_{1,2}$ as in (1) until merging into a unique cycle $\rc_{1,2,3}$
		with 6 vertices. This requires $4$ steps. 
		
		\item[]\hspace{1cm}$\vdots$
		
		\item[($k$)] Break the cycle $\rc_{1,\dots,k-1}$ and connect it to $\rc_{k}$. This requires $1$ step.
		Enlarge $\rc_{k}$ by collecting the elements of $\rc_{1,\dots,k-1}$ as before. 
		This requires $2(k-1)$ steps. 
		This final cycle is Hamiltonian, hence a \ac{NE} (see \cref{re:Hamilton}).
	\end{itemize}
	In total we have $\sum_{i=1}^{k-1} 2i=k(k-1)$ steps of type \ref{it:deviation-D1} and $k-1$ steps of type \ref{it:deviation-D2},
	which altogether give exactly $(k+1)(k-1)=k^2-1=\frac{1}{4}\nPlayers^2-1$ steps.
\Halmos
\endproof

\begin{remark}
	Notice that the trivial upper bound for the  length of a path in the meta-graph of the game is $\abs{\actions}$. For a directed graph with minimum out-degree 2, that is,  when players can really act strategically, the size of $\actions$ is  exponentially large with $\abs{\actions}\geq 2^{\nPlayers}$,  which is far worse than the quadratic bound established above.
\end{remark}

%
% Section ------------------------------------------
%

\section{The stochastic buck-passing game.}
\label{se:stochastic}
%\subsection{The game}
We now extend the deterministic model of \cref{se:deterministic} by allowing players to pass the buck at random to some neighbor. 
Specifically,  the strategy set of  player $\play\in\vertices$  is now the simplex of probabilities over $\neighborsout_{\play}$. 
This is isomorphic to the set $\simplex_{\play}$  of probability vectors $\transitvec_{\play}$ on $\vertices$ such that $\transit_{\play\playalt}=0$ for all $\playalt \not\in \neighborsout_{\play}$. 
Moreover,  $\simplex=\times_{\play\in\vertices}\simplex_{\play}$ denotes the set of strategy profiles.
With a slight abuse of notation, the same symbol $\transitmatrix$ will denote the strategy profile $\parens{\transitvec_{1},\dots,\transitvec_{\nPlayers}}\in\simplex$ and the stochastic matrix $[\transit_{\play\playalt}]_{\play,\playalt\in\vertices}$ whose rows are $\transitvec_{1},\dots,\transitvec_{\nPlayers}$. 
The symbol $\graph_{\transitmatrix}$ denotes the induced \emph{weighted} directed graph with vertices in  $\vertices$, edges in $\edges_{\transitmatrix}\coloneqq\braces{(\play,\playalt):\transit_{\play\playalt}>0}$, and weights $\transit_{\play\playalt}$.

At the start of the game each player chooses a strategy $\transitvec_{\play}\in\simplex_{\play}$. 
At time $\per=0$  the buck is given  to a player $\play$ drawn at random according to the measure $\initial$. 
This player then passes the buck to a random neighbor sampled according 
to $\transitvec_{\play}$,  and so on. 
Rigorously, the process is a time-homogeneous Markov chain $(\markov_{\per})_{\per\ge0}$ with initial measure $\initial$ and transition matrix $\transitmatrix$, that is,
\begin{equation}
\label{eq:Markov-chain}
\prob_{\transitvec}(\markov_{0}=\play)=\initial_{\play} \quad\text{and}\quad\prob_{\transitvec}(\markov_{\per+1}=\playalt\mid\markov_{\per} = \play)=\transit_{\play\playalt},
\end{equation}
where $\prob_{\transitvec}$ is the probability measure induced by the strategy profile $\transitvec$.
The cost $\cost_{\play}\colon\simplex\to\reals$ of player $\play$ is again defined by \cref{eq:SBPG-cost}, but this time the expectation is taken  with respect to both the initial measure $\initial$ and the transition matrix $\transitmatrix$.
We call $\game(\graph,\initial,\simplex)$ a \acfi{SBPG}\acused{SBPG}, and we write $\NE(\simplex)$ for the  set of its \aclp{NE}.

The analysis of \acp{SBPG} requires some  standard concepts in the theory of Markov chains.
We recall that a \emph{recurrent class} in $\graph_{\transitmatrix}$ is a strongly connected 
component  $\rc$ which is maximal by inclusion.
For each $\transitmatrix$ the vertex set can be partitioned into a finite disjoint union of sets, that is,
\begin{equation}\label{eq:partition}
\vertices=\vertices_{\transitmatrix}^{0}\cupdot\rc_{\transitmatrix}^{1}\cupdot\dots\cupdot\rc_{\transitmatrix}^{\numrc(\transitmatrix)},
\end{equation}
where, for all $\countrc\in\{1,\dots\numrc(\transitvec) \}$, $\rc_{\transitmatrix}^\countrc$ is a recurrent class in $\graph_{\transitmatrix}$,  $\vertices_{\transitmatrix}^{0}$ is the set of \emph{transient vertices} that do not belong to any recurrent class, and $\cupdot$ indicates the disjoint union. 

The restriction of the chain to each $\rc_{\transitmatrix}^{\countrc}$ is itself an \emph{irreducible chain} which supports a 
unique stationary measure $\stationary_{\transitmatrix}^{\countrc}$. For convenience we extend this measure to the 
full vertex set by setting  $\stationary_{\transitmatrix}^{\countrc}(\play)=0$ for $\play\not\in \rc_{\transitmatrix}^{\countrc}$.
These measures are characterized by the classical ergodic theorem.

\begin{theorem}[Ergodic Theorem]
	\label{th:ergodic}
	Consider an irreducible Markov chain $(\markov_{\per})_{\per\ge0}$  on a finite state space $\vertices$, with transition matrix $\transitvec$ and initial distribution $\initial$. Let $\stationary_{\transitvec}$ be the unique stationary measure, and $\horizon_{\play}=\inf\braces{\per >0\vert \markov_{\per}=\play}$ the hitting time of state $\play$. Then
	\begin{equation}\label{eq:stdis}
	\expect
	\bracks*{\lim_{\horizon\to\infty}\frac{1}{\horizon}\sum_{\per=1}^{\horizon}\mathds{1}_{\braces{\markov_{\per}=\play}}}=\stationary_{\transitvec}(\play)
	=\frac{1}{\expect
		[\horizon_{\play}\mid \markov_{0}=\play]}.
	\end{equation}
\end{theorem}
The proof of \cref{th:ergodic} can be found, for instance, in \cite[Theorem~C.1 and Proposition~1.19]{LevPer:AMS2017}.

For a general Markov chain (not necessarily irreducible), starting from any state $\playalt\in\vertices$ the chain $\markov$ will  eventually be absorbed in a recurrent class $\rc_{\transitvec}^\countrc$ with probability
\begin{equation}
\label{eq:absorb}
\absorb_{\transitmatrix}^{\playalt\to \countrc}\coloneqq \prob\parens{\text{there exists }\horizon\in\naturals\text{ such that  }\markov_{\horizon}\in \rc_{\transitvec}^\countrc \mid \markov_{0}=\playalt},
\end{equation}
so that the total probability with which the buck is absorbed in $\rc_{\transitvec}^{\countrc}$ is 
\begin{equation}
\label{eq:total-mass}
\initial_{\transitvec}^{\countrc}=\sum_{\playalt\in\vertices}\initial_{\playalt}\, \absorb^{\playalt\to \countrc}_{\transitvec}.
\end{equation}
Thus, applying \cref{th:ergodic} on each recurrent class  $\rc_{\transitvec}^{\countrc}$, the cost  for player $\play$  in \cref{eq:SBPG-cost} can be finally expressed as
\begin{equation}
\label{eq:cost-general}
\cost_{\play}(\transitvec)=\sum_{\countrc=1}^{\numrc(\transitvec)}\initial_{\transitvec}^{\countrc}\, \stationary_{\transitvec}^{\countrc}(\play).
\end{equation}
Note that this is  equivalent to \cref{eq:DBPG-cost} when $\transitvec$ is a deterministic strategy profile.

\begin{remark}
\label{re:mixed}
	A \ac{SBPG} is \emph{not} the mixed extension of the deterministic game \ac{DBPG}. 
	Given  the cost function  in \cref{eq:SBPG-cost}, a probability vector $\transitvec_{\play}\in\simplex_{\play}$ is not a mixed strategy for \ac{DBPG}. 
	The difference is that in a \ac{SBPG}, each time a player receives the buck, a new neighbor is drawn at random, whereas in the mixed extension of \ac{DBPG} this random neighbor is drawn at the beginning of the game and kept fixed thereafter. 
	To illustrate the difference, let $\graph$ be  the complete graph  on three vertices $\vertices=\braces{1,2,3}$ and consider the strategy profile 
	(see \cref{fig:mixed-strategy}):
	\begin{equation}\label{eq:strategySBPG}
	\transitvec_{1}=\parens*{0,\frac{1}{2},\frac{1}{2}}, \quad
	\transitvec_{2}=\parens{1,0,0}, \quad 
	\transitvec _{3}=\parens{0,1,0}.
	\end{equation}
	
	\begin{figure}[h]
		\FIGURE
		{\includegraphics[width=4cm]{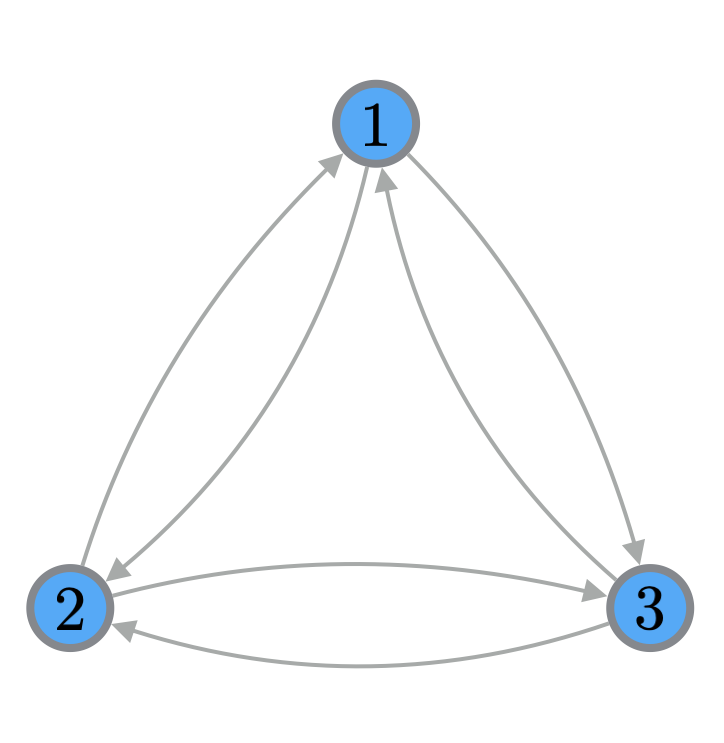}
		\qquad\qquad
		\includegraphics[width=4cm]{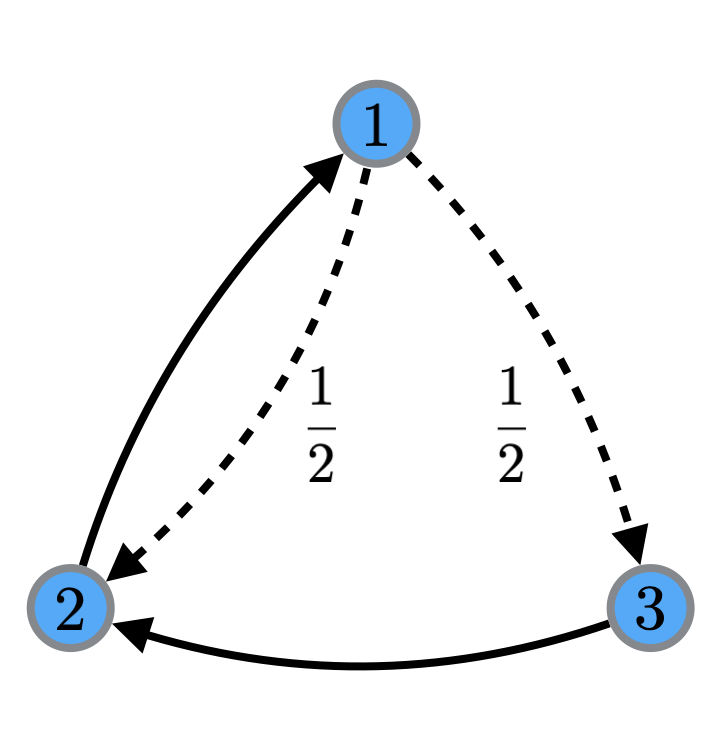}}
		{Left: A complete graph with $3$ vertices. Right: The strategy \eqref{eq:strategySBPG}.
		\label{fig:mixed-strategy}}
	{}
	\end{figure}
	
	\noindent
	Using \eqref{eq:cost-general}, we can see that in the \ac{SBPG} the profile  $\transitmatrix$ has a cost vector equal to the stationary distribution it induces, that is,
	\begin{equation*}\label{eq:stationaryMDBP}
	\parens*{\frac{2}{5},\frac{2}{5},\frac{1}{5}}.
	\end{equation*}
	If instead we take  \eqref{eq:strategySBPG}  as a mixed strategy in the \ac{DBPG}, then the cost for player $1$ is
	\begin{equation*}\label{eq:cost1SBPG}
	\begin{aligned}
	\cost_{1}(\transitvec)
	&=
	\frac{1}{2}\,\expect\bracks*{\lim_{\horizon\to\infty}\frac{1}{\horizon}\sum_{\per=1}^{\horizon}\buck_{1,\per}(2,1,2)}+
	\frac{1}{2}\,\expect\bracks*{\lim_{\horizon\to\infty}\frac{1}{\horizon}\sum_{\per=1}^{\horizon}\buck_{1,\per}(3,1,2)}\\
	&=
	\frac{1}{2}\frac{1}{2}+\frac{1}{2}\frac{1}{3}=\frac{5}{12}\neq\frac{2}{5}.
	\end{aligned}
	\end{equation*}
\end{remark}

\subsection{Relation between the deterministic and stochastic games.}
Even if  \ac{SBPG} is not the mixed extension of \ac{DBPG}, we next show that---as it happens for mixed extensions---equilibria for the deterministic game are preserved in the stochastic setting.
In line with the notation introduced in \cref{se:deterministic},  $\actions_{\play}$ denotes the set of extreme points of $\simplex_{\play}$, so a strategy in $\actions_{\play}$ is a degenerate measure and corresponds to choosing an out-neighbor with probability $1$.

\begin{proposition}\label{pr:MDP}
	For each  \acl{BPG} we have $\NE(\actions)\subseteq\NE(\simplex)$.
\end{proposition}

\proof{Proof.}
	By contradiction, suppose there exists some $\actprof\in\NE(\actions)$ such that $\actprof\not\in\NE(\simplex)$. 
	Notice that
	\begin{equation}\label{eq:cost-exp}
	\begin{split}
	\expect_{\transitvec}[\horizon_{\play} \mid \markov_{0}=\play]
	&=1+\sum_{\playalt\in\vertices}\transit_{\play\playalt}\expect_{\transitvec}[\horizon_{\play} \mid \markov_{0}=\playalt]\\
	&=1+\sum_{\playalt\in\vertices}\transit_{\play\playalt}\expect_{\transitvec_{-\play}}[\horizon_{\play} \mid \markov_{0}=\playalt],
	\end{split}
	\end{equation}
	with the obvious meaning of the symbols.
	In particular, if we consider a deterministic profile $\actprof\in\actions$, then 
	\begin{equation}\label{eq:cost-exp-det}
	\expect_{\actprof}[\horizon_{\play} \mid \markov_{0}=\play]=1+\expect_{\actprof_{-\play}}[\horizon_{\play} \mid \markov_{0}=\act_{\play}].
	\end{equation}
	Consider first the case where $\numrc(\actprof)=1$. 
	In this setting, if $\actprof\in\NE(\actions)$, then, thanks to \eqref{eq:stdis}, for every $\play\in\vertices$, 
	\begin{equation}\label{eq:pure-Nash-arg-max}
	\act_{\play}\in\argmax_{\playaltalt\in\neighborsout_{\play}}\,\expect_{\actprof_{-\play}}[\horizon_{\play}\mid \markov_{0}=\playaltalt].
	\end{equation}
	Hence, by \eqref{eq:cost-exp}, if $\actprof\not\in\NE(\simplex)$, there exists  $\play\in\vertices$ and some $\transitvec_{\play}\in\simplex_{\play}$ such that
	\begin{equation}
	\sum_{\playalt\in\vertices}\transit_{\play\playalt}\expect_{\actprof_{-\play}}[\horizon_{\play} \mid \markov_{0}=\playalt]>\expect_{\actprof_{-\play}}[\horizon_{\play} \mid \markov_{0}=\act_{\play}].
	\end{equation}
	In particular, there exists some $\playalt\neq\act_{\play}$ such that
	\begin{equation}
	\expect_{\actprof_{-\play}}[\horizon_{\play} \mid \markov_{0}=\playalt]>\expect_{\actprof_{-\play}}[\horizon_{\play} \mid \markov_{0}=\act_{\play}].
	\end{equation}
	In this case, $\actalt_{\play}=\playalt$ is a profitable deviation for player $\play$ for the \acl{DBPG}. 
	This contradicts the assumption that $\actprof\in\NE(\actions)$.	
	
	If $\actprof$ induces more than a single recurrent class and $\actprof\in\NE(\actions)$, then, whenever $\play$ lies in a cycle $\rc_\countrcalt$ of $\graph_{\actprof}$ having $\initial_{\actprof}^{\countrcalt}>0$, all her out-neighbors are in the same unicycle. 
	Hence, by the same argument as above, there is no randomized profile $\transitvec=(\transitvec_\play,\actprof_{-\play})$ such that  $\stationary^{\countrcalt}_{\transitvec}(\play)<\stationary^{\countrcalt}_{\actprof}(\play)$. On the other hand, if $\initial^{\countrcalt}_{\actprof}=0$ then, clearly, player $\play$ has no improving deviation.  Therefore, also in the case of multiple recurrent classes we have $\actprof\in \NE(\actions)\Rightarrow\actprof\in\NE(\simplex)$.
\Halmos
\endproof

The \acl{SBPG} gives each player a richer set of strategies, so it is not surprising that we get a larger set of equilibria. 
\cref{ex:di-wheel} below shows that the set of equilibria may be strictly larger in the stochastic setting, and 
illustrates that some players may be favored by the larger strategy sets, whereas others will be affected negatively.

\begin{example}\label{ex:di-wheel}
	Consider a graph consisting of a uni-directional cycle with vertices $\play=1,\ldots,\nPlayers-1$ 	plus a central player $\nPlayers$ connected bi-directionally to all the other vertices (see \cref{poa-cex2}). 
	For each player $\play$ on the outer cycle, $\play_{-}$ and $\play_{+}$ denote, respectively, the predecessor and successor along the cycle.
	
\begin{figure}[h]
	\FIGURE
	{\includegraphics[width=5cm]{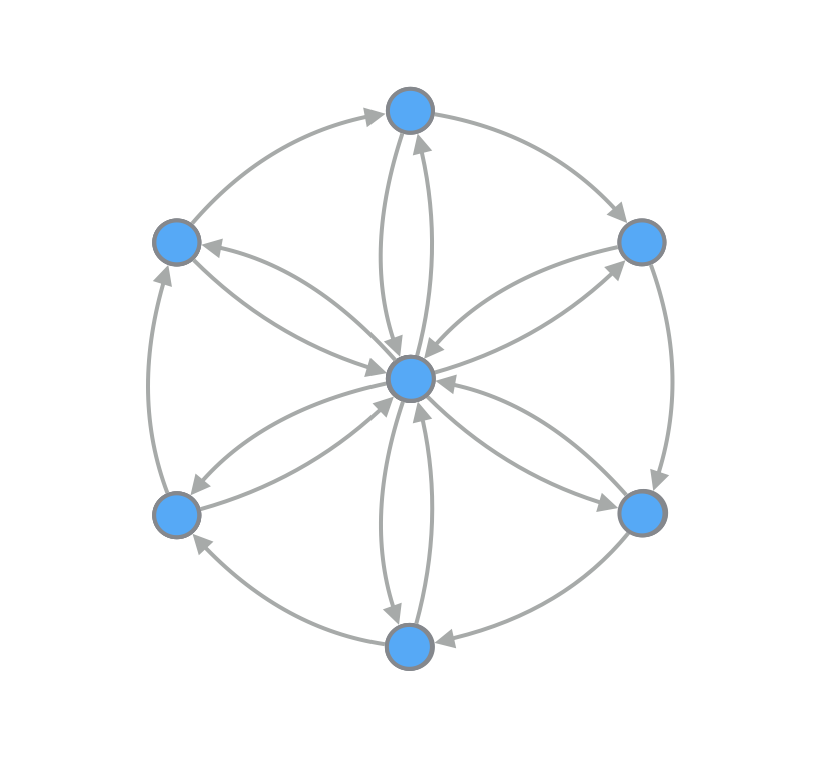}}
	{A clockwise wheel, with $6$ outer vertices plus a central player.\label{poa-cex2}}
	{}
	\end{figure}
	
	Consider first the \ac{DBPG}. 
	Once the central player has designated an out-neighbor $\play$, it is a dominant strategy for the latter to forward the buck  along the cycle  to her out-neighbor $\play_+$. 
	The same holds for all subsequent players along the cycle, except for player $\play_-$, for whom it is a dominant strategy to return the buck to the central player. 
	This yields a Hamiltonian cycle. 
	There are exactly $\nPlayers-1$ such cycles, one for each $\play$, and these are exactly the equilibria in $\NE(\actions)$ (see \cref{re:Hamilton}). 
	In each of these equilibria each player pays $1/\nPlayers$.
	
	Now consider the strategy profile $\transitvec\in\simplex$ in the \ac{SBPG} in which the central player sends the buck uniformly at random with probability $1/(\nPlayers-1)$ to each vertex in the outer ring, whereas any other player $\play$ plays a deterministic strategy  sending the buck to the next player $\play_{+}$ along the cycle. 
	We claim that  $\transitvec\in\NE(\simplex)$. The central player is transient and pays $0$, so that she has no profitable deviation. 
	Each player $\play$ on the outer ring pays $1/(\nPlayers-1)$. 
	If she deviates to a stochastic strategy by sending to player $\play_+$ with probability $\proba$ and to the central player with probability $(1-\proba)$, then her expected return time is
	\begin{align*}
	\expect_{\transitvec}[\horizon_{\play} \mid \markov_{0}=\play]
	&=1+\proba\,\expect_{\transitvec}[\horizon_{\play} \mid \markov_{0}=\play_+]
	+(1-\proba)\,\expect_{\transitvec}[\horizon_{\play} \mid \markov_{0}=\nPlayers]\\
	&=1+\proba\,(\nPlayers-2) +(1-\proba)\frac{1}{\nPlayers-1}(1+2+\cdots+(\nPlayers-1))\\
	&=1+\proba\,(\nPlayers-2) +(1-\proba)\frac{\nPlayers}{2}.
	\end{align*}
	For $\nPlayers\geq 5$ this expression is strictly increasing with $\proba$ so that, according to \eqref{eq:stdis}, the minimum cost is attained for $\proba=1$, which proves that $\transitvec$ is indeed a \ac{NE}.
	In this equilibrium the central player pays 0 and is better off than in the \ac{DBPG}, whereas
	all the other players are worse off, since their cost is now $1/(\nPlayers-1)$ rather than $1/\nPlayers$.
	
\end{example}

%
% Section ------------------------------------------
%

\section{Fairness of equilibria.}
\label{se:fairness}

%\subsection{Measures of fairness}\label{suse:inefficiency}

As seen in \cref{ex:multiple-eq}, a \acl{BPG} may have several equilibria and in some of them the total cost is very unevenly spread among players. 
We want to compare---in terms of fairness---the equilibrium cost vectors with the optimum cost vectors that could be achieved by a benevolent social planner, whose goal is to minimize disparity in the way players are treated. 
To this end we adopt a Rawlsian criterion, \citep[see][]{Raw:HUP2009}, and we define the social cost of a strategy profile as the cost incurred by the player who pays the most
\begin{equation}
\label{eq:SC}
\SC(\actprof)=\max_{\play\in\vertices}\cost_{\play}(\actprof),
\end{equation}
so that minimizing this social cost corresponds somehow to maximizing fairness among the players of the game. 

Equilibrium and optimum costs are usually compared in terms of efficiency.
Typically the social cost function  is taken as  the sum of the costs incurred by all players.  
The standard measures of efficiency in games are the \acfi{PoA}\acused{PoA}, i.e., the ratio between the social cost of the worst equilibrium and the minimum social cost, 
and the \acfi{PoS}\acused{PoS}, i.e., the ratio between the social 
cost of the best equilibrium and the minimum social cost. 
Explicitly, if $\SC\colon\actions\to\reals$ is the social cost function for a \acl{DBPG}, we define
\begin{align}
\label{eq:PoA-D}
\PoA(\actions)&\coloneqq\frac{\max_{\actprof\in\NE(\actions)}\SC(\actprof)}{\min_{\actprof\in\actions}\SC(\actprof)},\\
\label{eq:PoS-D}
\PoS(\actions)&\coloneqq\frac{\min_{\actprof\in\NE(\actions)}\SC(\actprof)}{\min_{\actprof\in\actions}\SC(\actprof)}.
\end{align}
For \aclp{SBPG},  $\SC\colon\simplex\to\reals$  and 
\begin{align}
\label{eq:PoA-S}
\PoA(\simplex)&\coloneqq\frac{\max_{\transitvec\in\NE(\simplex)}\SC(\transitvec)}{\min_{\transitvec\in\simplex}\SC(\transitvec)},\\
\label{eq:PoS-S}
\PoS(\simplex)&\coloneqq\frac{\min_{\transitvec\in\NE(\simplex)}\SC(\transitvec)}{\min_{\transitvec\in\simplex}\SC(\transitvec)}.
\end{align}

Since the \acl{BPG} is a constant sum game, efficiency is not an issue: if social cost is the sum of the player costs, all strategy profiles are equally efficient. 
On the other hand, by using the social cost function in \eqref{eq:SC}, the 
\ac{PoA} and \ac{PoS} can be used to measure fairness. 
The smaller the \ac{PoA} (\ac{PoS}), the fairer the worst (best) equilibrium.
We are not the first to use social cost functions that are not the sum of individual costs \citep[see for instance][]{KouPap:CSR2009,KouPap:STACS1999,Vet:FOCS2002,MavMonPap:Springer2008,FouSca:MOR2019}.

The next proposition establishes tight  bounds for both the \ac{PoA} and the \ac{PoS}. 

\begin{proposition}\label{pr:PoA-PoS-mod}
	For any \acl{BPG} $\game(\graph,\initial)$ we have
	\begin{enumerate}[\upshape(a\upshape)]
		\item\label{it:pr:PoS-mod}
		$\PoS(\actions)\leq\PoA(\actions)\leq \nPlayers/2$.
		\item\label{it:pr:PoA-mod}
		$\PoS(\simplex)\leq\PoA(\simplex)\le \nPlayers/2$. 
	\end{enumerate}
	Moreover, there are instances where $\PoS(\actions)=\PoS(\simplex)=\nPlayers/2$ and all these inequalities are satisfied as equalities.
\end{proposition}

\proof{Proof.}
	By definition, $\PoS$ is always smaller than $\PoA$, so that it suffices to establish the upper bound of $\nPlayers/2$ in both \ac{DBPG} and \ac{SBPG}.
	In both cases the worse that can happen to a player is to receive the buck every other period, so that, for any possible strategy profile, no player pays more than $1/2$. On the other hand, since the sum of costs over all players is $1$, the minimum social cost in both settings is at least $1/\nPlayers$.
	This implies a bound of $\nPlayers/2$ in both the deterministic and stochastic cases.
	
	To show that these bounds can be reached consider a graph consisting of two disjoint directed cycles with $\nPlayers-2$ and $2$ players, respectively, and only one pivot player in the longest cycle who has an additional link connecting to the 2-cycle. See \cref{fig:poa-cex1}.
	
	\begin{figure}[h]
		\FIGURE
		{\includegraphics[width=6.5cm]{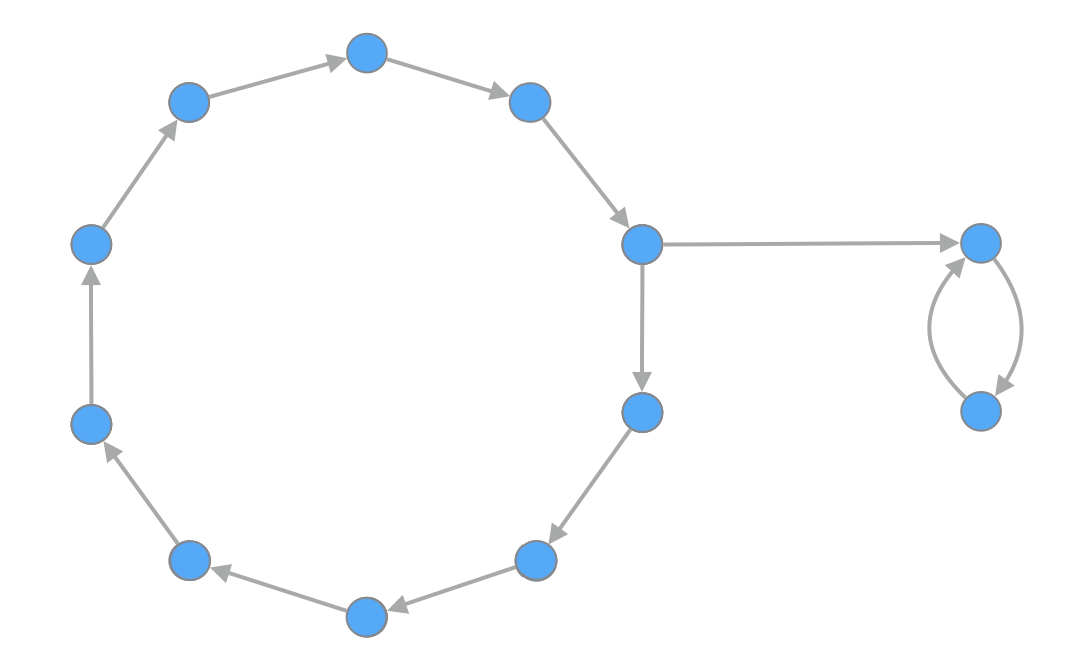}}
		{A graph with $\PoS=\PoA=\nPlayers/2$.\label{fig:poa-cex1}}
		{}
	\end{figure}
	
	This pivot player is the only one who can randomize by sending the buck to the 2-cycle with probability $\proba$ and following the long cycle with probability
	$(1-\proba)$. 
	For a uniform initial measure $\initial$,  the unique strategy that minimizes
	the social cost is the deterministic one which sets $\proba=0$ inducing a social cost of $1/\nPlayers$ (perfectly fair).
	This is not an equilibrium though, since the pivot player may deviate to $\proba\in(0,1]$ in the SBPG and to $\proba=1$ in the \ac{DBPG}, and in all these equilibria the buck is absorbed in the $2$-cycle, whose players pay $1/2$ each. 
	Hence in both cases we get $\PoS(\actions)=\PoS(\simplex)=\nPlayers/2$.
\Halmos
\endproof

A natural question is how $\PoA$ and $\PoS$ change from the deterministic to the stochastic settings. 
The comparison is not straightforward  since the  reference baseline set by the minimal social cost may be different
in both cases. Moreover, as seen in \cref{ex:di-wheel}, even if the graph is Hamiltonian with the same baseline
\[
\min_{\transitvec\in\simplex}\SC(\transitvec)= \min_{\actprof\in\actions}\SC(\actprof)=1/\nPlayers,
\]
we may still have
\begin{equation*}\label{eq:PoA-ineq}
\PoA(\simplex)\geq\frac{\nPlayers}{\nPlayers-1}>\PoA(\actions)=1.
\end{equation*}

\begin{proposition}\label{pr:mod12} For any \acl{BPG} $\game(\graph,\initial)$ we have
	\begin{enumerate}[\upshape(a\upshape)]
		\item\label{it:pr:PoS-mod1} $\min_{\transitvec\in\simplex}\SC(\transitvec)\leq \min_{\actprof\in\actions}\SC(\actprof)$
		\item\label{it:pr:PoS-mod2} $\PoA(\actions)\leq\PoA(\simplex)$, 	possibly with strict inequalities.  
	\end{enumerate}
\end{proposition}

\proof{Proof.}
	Part \ref{it:pr:PoS-mod1} follows directly by noting that each $\actprof\in \actions$ is equivalent to a 
	deterministic strategy in $\simplex$, and then \ref{it:pr:PoS-mod2} follows from the inclusion
	$\NE(\actions)\subseteq\NE(\simplex)$ in \cref{pr:MDP}.
	To see that the inequalities may be strict consider the graph in \cref{fig:Bi-Cycle}
	composed of two disjoint cycles with 2 vertices each, plus a transient vertex
	connected to both cycles. 
	\begin{figure}[h]
		\FIGURE
		{\includegraphics[width=9.5cm]{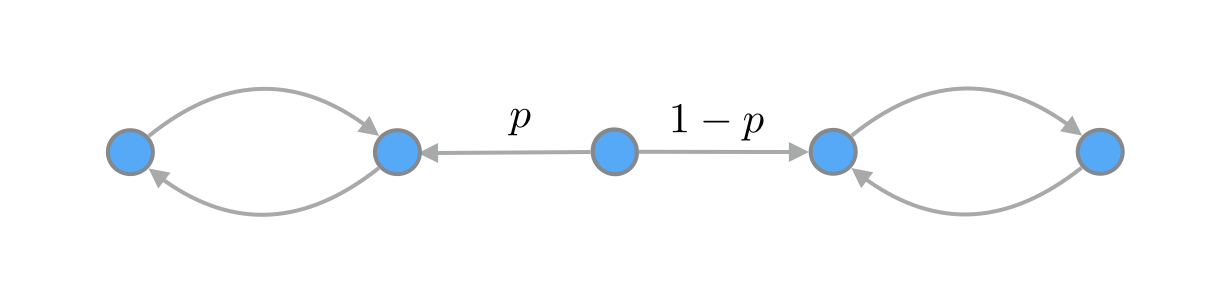}}
		{A graph with $\displaystyle\min_{\transitvec\in\simplex}\SC(\transitvec)<\min_{\actprof\in\actions}\SC(\actprof)$. The initial measure $\initial$ is concentrated on the central vertex.\label{fig:Bi-Cycle}}
		{}
	\end{figure}
	\noindent
	If we initially assign the buck to this transient vertex with probability one, then
	the minimum social cost over $\simplex$ attained with $p=\frac{1}{2}$, whereas over 
	$\actions$ the optimum is attained for $p=0$ and $p=1$. Hence
	\begin{align*}
	\min_{\transitvec\in\simplex}\SC(\transitvec)=\frac{1}{4}&<\min_{\actprof\in\actions}\SC(\actprof)=\frac{1}{2},
	\end{align*}
	\begin{align*}
	\PoA(\actions)=1&<\PoA(\simplex)=2. \Halmos
	\end{align*}
\endproof

If \cref{pr:mod12}\ref{it:pr:PoS-mod1} holds with equality, the inclusion $\NE(\actions)\subseteq\NE(\simplex)$
implies $\PoS(\simplex)\leq\PoS(\actions)$. However, as shown by \cref{ex:PoAPoS_order} in the next section,
in general there is no order between $\PoS(\simplex)$ and $\PoS(\actions)$. 
The following special case presents a situation where the deterministic and stochastic \acl{BPG} have
\acl{PoS} equal to 1.
\begin{proposition}
	If $\graph$ is a disjoint union of strongly connected components, then for every \acl{BPG} $\game(\graph,\initial)$ we have $\PoS(\actions)=\PoS(\simplex)=1$.
\end{proposition}

\proof{Proof.}
	It suffices to show that the optimal social cost over $\simplex$ is attained at a deterministic strategy $\actprof\in\actions$ which is also an equilibrium,
	namely $\actprof\in\NE(\actions)\subseteq\NE(\simplex)$.
	
	Let $\scc_1,\dots\scc_\numrc$ be the strongly connected components in $\graph$ and choose  a collection of longest cycles $\rc_1,\dots,\rc_\numrc$, one in each component.
	Consider the strategy profile $\actprof$ induced by  these cycles, where all the other players are free riders. 
	By construction, none of these cycles can be destroyed nor extended, so that $\actprof$ is a prior-free \ac{NE} for the \ac{DBPG}, hence also for the \ac{SBPG} by virtue of \cref{pr:MDP}. 
	It remains to show that, for every initial measure $\initial$, this deterministic strategy profile $\actprof$ minimizes the social cost over $\simplex$ (hence also over $\actions$). 
	Indeed, let $\initial_\countrc$ denote the initial mass of the component $\scc_\countrc$. Being $|\rc_\countrc|$ the length of the maximal cycle in $\scc_\countrc$, for every $\transitvec\in\simplex$ and each player $\play\in \rc_\countrc$, we have $\expect_{\transitvec}[\horizon_{\play} \mid \markov_{0}=\play]\leq |\rc_\countrc|$, so that \eqref{eq:stdis}  yields $\cost_\play(\transitvec)\geq \initial_\countrc/\abs{\rc_\countrc}=\cost_\play(\actprof)$.
	It follows that $\SC(\transitvec)\geq\SC(\actprof)$ completing the proof.
\Halmos
\endproof

%
% Subsection ------------------------------------------
%
\subsection{Some examples.}
\label{suse:examples}

The following examples show that, depending on the structure of the graph $\graph$, the inequalities
in \cref{pr:PoA-PoS-mod} can be tight or not. The first two examples concern respectively the cases of 
complete graphs and cycles, with a uniform initial measure $\initial$. These graphs are symmetric in 
the sense that $(\play,\playalt)\in\edges$ iff $(\playalt,\play)\in\edges$, and they have a Hamiltonian
cycle which is both an optimal profile and a Nash equilibrium so that the optimal social cost is 
$1/\nPlayers$ and $\PoS(\actions)=\PoS(\simplex)=1$.

\begin{figure}[h]
	\FIGURE
	{\includegraphics[width=8.5cm]{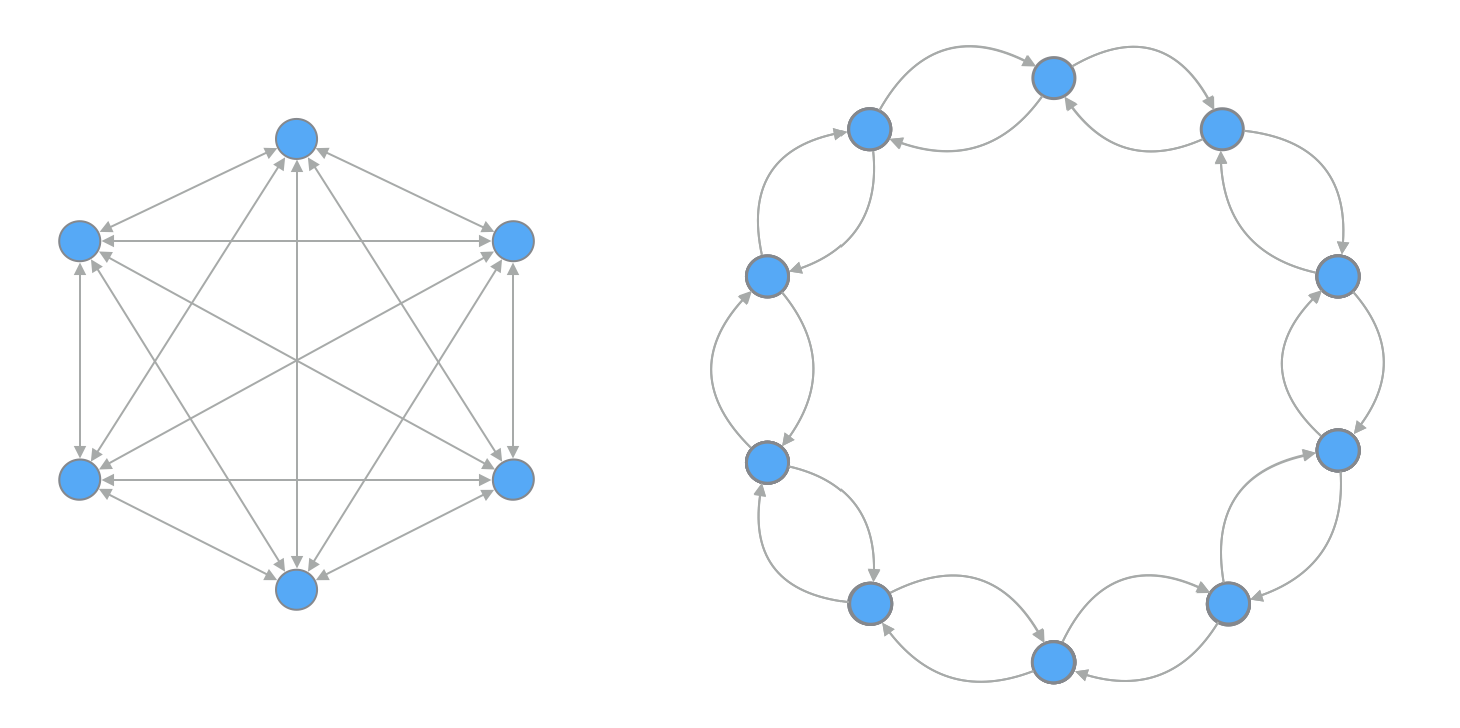}}
	{The complete graph $\completegraph_{6}$ and the bi-directional cycle $\cyclegraph_{10}$. For ease of representation in the first picture we used bidirectional arrows instead of drawing two arrows for each couple of vertices. These two graphs admit a Hamiltonian cycle.
	\label{fig:undirected-graphs}}
{}
\end{figure}

\begin{example}
	For the \acl{BPG} on the complete graph $\completegraph_{\nPlayers}$ with $\nPlayers$ vertices, the Nash equilibria (both in the deterministic and stochastic versions)
	are exactly the profiles $\actprof\in\actions$ such that $\graph_{\actprof}$ is a Hamiltonian cycle, and therefore 
\[
\PoS(\actions)=\PoA(\actions)=1=\PoS(\simplex)=\PoA(\simplex).
\]
\end{example}

\begin{example}\label{ex:cycle}
	In the bi-directional cycle $\cyclegraph_{\nPlayers}$ on $\nPlayers$ vertices the game may give rise to very unfair equilibria.
	Notice that the strategy of each player reduces to a choice between her left and right neighbors, and all cycles are either of length 2 or $\nPlayers$. 
	As noted before, a player will never pay more than 1/2. 
	However, the equilibrium described in \cref{fig:prior-free-ce} features exactly two players on a cycle of length 2 and each one pays exactly 1/2, so that in this case 
\[
\PoS(\actions)=\PoS(\simplex)=1<\PoA(\actions)=\PoA(\simplex)=n/2.
\]
\end{example}

\begin{example}
	In the previous examples we had either $\PoS(\actions)=\PoS(\simplex)=1$ or
	$\PoA(\actions)=\PoA(\simplex)=\nPlayers/2$. For the graph in \cref{fig:PoA_PoS},
	with a uniform initial measure $\initial$, these values are bounded away from these extremes.
	
	\begin{figure}[h]
		\FIGURE
		{\includegraphics[width=9.5cm]{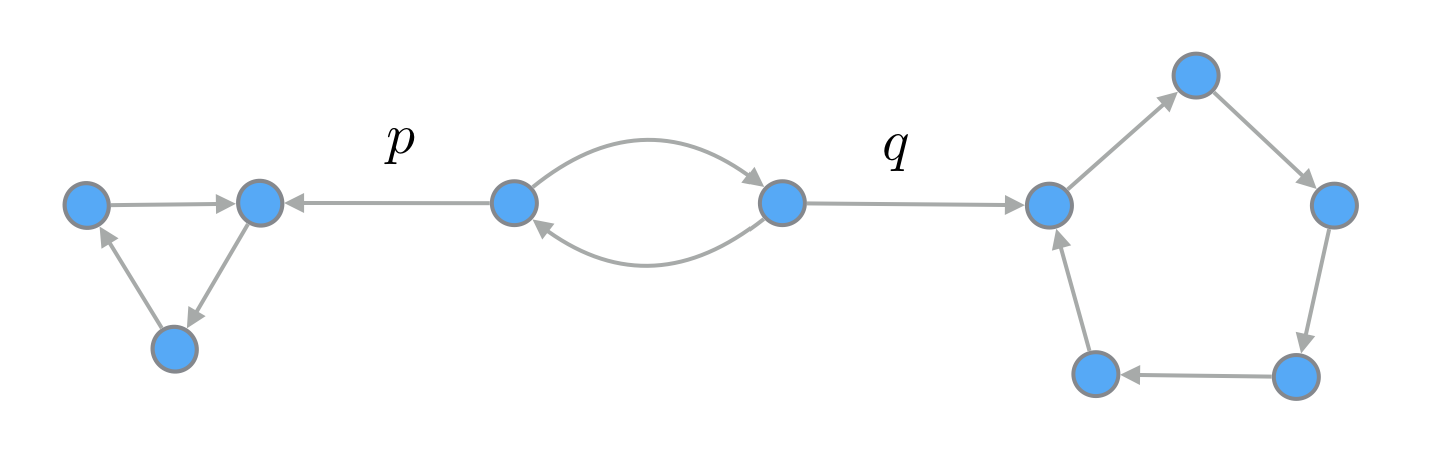}}
		{A network with $1<\PoS<\PoA<\frac{\nPlayers}{2}$. Here we consider $\initial$ as the uniform measure on the vertex set. \label{fig:PoA_PoS}}
		{}
	\end{figure}
	
	Indeed, notice that only the two players on the central cycle have the possibility to randomize.
	The optimal strategy is attained for $p=q=0$ with
\[
\min_{\transitvec\in\simplex}\SC(\transitvec)= \min_{\actprof\in\actions}\SC(\actprof)=\frac{1}{10}.
\]
	However, this is not an equilibrium and $\NE(\simplex)$ is precisely characterized by $p+q>0$.
	The best equilibrium is achieved with $p=0, q=1$, and the worse with  $p=1, q=0$. These are in fact deterministic equilibria in $\NE(\actions)$, so that
\[
1<\PoS(\actions)=\PoS(\simplex)=\frac{7}{5}<\frac{5}{3}=\PoA(\actions)=\PoA(\simplex)<\frac{\nPlayers}{2}.
\]
\end{example}

\begin{example}  \label{ex:PoAPoS_order}
	In general there is no order between $\PoS(\actions)$ and $\PoS(\simplex)$, and either one may be larger.
	Indeed, in \cref{fig:PoS<} the optimal cost  both for $\actions$ and $\simplex$ is $1/6$, attained with $p=q=0$. 
	The best equilibrium over $\actions$ is attained for $p=1,q=0$ (as well as $p=0,q=1$) with social cost $1/3$,
	whereas the best equilibrium over $\simplex$ is attained with $p=q=1/2$ for a social cost of $1/4$. 
	Hence $3/2=\PoS(\simplex)< \PoS(\actions)=2$.
	
	\begin{figure}[h]
		\FIGURE
		{
		\includegraphics[width=8cm]{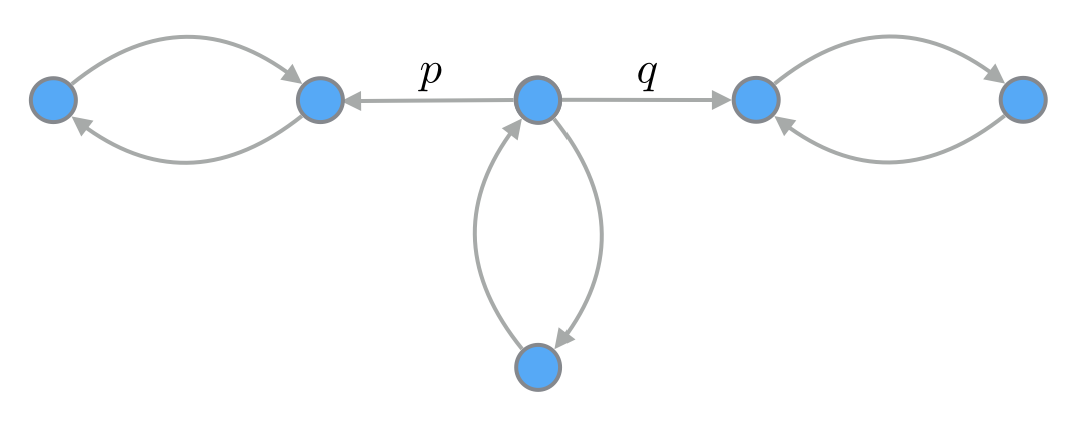}}
		{A \acl{BPG} with $\PoS(\simplex)<\PoS(\actions)$. The initial measure $\initial$ is uniform across all 6 vertices.\label{fig:PoS<}}
		{}
	\end{figure}
	
	Now, in \cref{fig:PoS>} the deterministic  social optimum is $1/3$ attained with $p=1,q=0$ and $r=1$. 
	This is also an equilibrium so that $\PoS(\actions)=1$. In the stochastic case the minimum cost
	is $1/6$ (attained with $p=q=1/2$ and $r=0$), whereas the best equilibrium is achieved for $r=1, p=3/4, q=1/4$, with social cost $1/4$ and then $\PoS(\simplex)=3/2>\PoS(\actions)$.
	
	\begin{figure}[h]
		\FIGURE
		{
		\includegraphics[width=8cm]{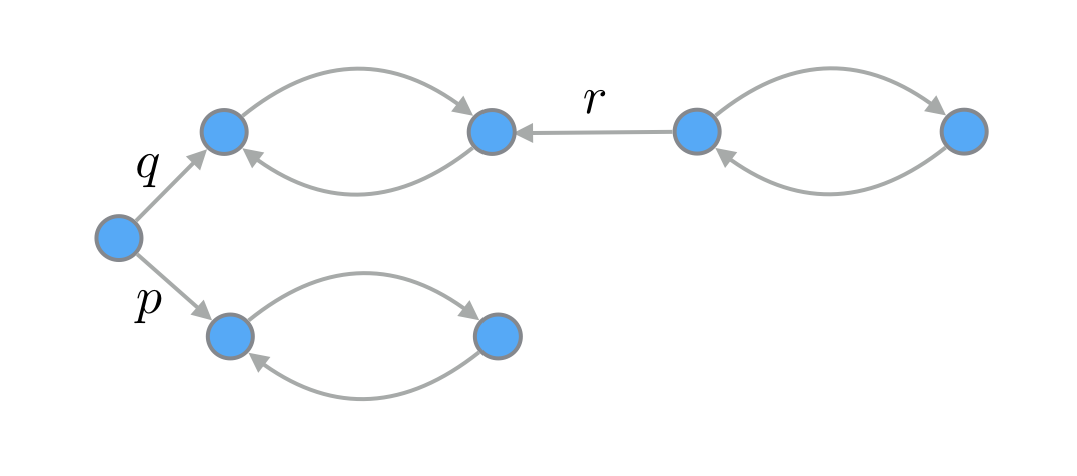}}
		{A  \acl{BPG} with $\PoS(\simplex)>\PoS(\actions)$. The initial measure $\mu$ assigns the buck to the leftmost vertex 
			with probability $\frac{2}{3}$, and to the rightmost vertex with probability $\frac{1}{3}$.\label{fig:PoS>}}
		{}
	\end{figure}
	
\end{example}

\section{Markov chains and spanning trees.}
\label{se:Markov-chains}
In \cref{se:gop} we will study the existence of equilibria for \aclp{BPG} when the strategy sets of the players are restricted to general subsets of the whole simplices of probabilities over out-neighbors. 
We will achieve this by extending the generalized potential function in \eqref{eq:potentialdef}, and this  requires some powerful tools in Markov chains.
In particular, we will exploit classical  results linking probability and graph theory.

\subsection{The Markov chain tree theorem revisited.}
\label{suse:Markov-chain-tree-theorem}

We begin by discussing the Markov chain tree formula, which characterizes the invariant measure of the chain.
To the best of our knowledge, this formula was first established by \citet{VenFre:UMN1970} for the case where 	$p_{ij}>0$ for all $j\neq i$. 
It was later extended to ergodic Markov chains by \citet{LeiRiv:IEEETIT1986}, and beyond by \citet{AnaTso:SPL1989}. 
See also \citet{AldFil:mono2002}, where the authors define such formula ``the most often rediscovered result in probability theory.''
For the reader's convenience we provide a short proof that relies on the analysis of spanning unicycles and unveils the connection between the formula and the \acl{BPG}. 
This proof will yield \cref{co:expected-cycle-length}, which provides an alternative expression and a useful bound for the potential.
We start by defining the objects needed to state the theorem, making use of the same vocabulary introduced in the description of the \acl{BPG}.

Consider a Markov chain on a finite state space $\vertices$
with transition matrix $\transitmatrix=[\transit_{\play\playalt}]_{\play,\playalt\in \vertices}$, and assume that it has a unique recurrent class $\rc\subseteq\vertices$, so that all vertices in $\vertices\setminus\rc$  are transient.
In this case there is a unique invariant measure $\stationary_{\transitvec}=(\stationary_{\transitvec}(\play))_{\play\in\vertices}$, with $\stationary_{\transitvec}(\play)>0$ iff $\play\in\rc$.

Consider the \emph{weight function} $\weight_{\transitvec}:2^{\edges}\to[0,1]$ defined as follows: for every $\event\subset\edges$
\begin{equation}
\label{eq:weight}
\weight_{\transitvec}(\event)=\prod_{(\play,\playalt)\in\event}\transit_{\play\playalt}.
\end{equation}

Note that each strategy profile $\actprof\in \actions$ in the deterministic buck-passing game
can be identified with the subset of induced edges $\edges_{\actprof}$. Thus---with a slight abuse of notation and only for this section---we identify $\actions$ with the family of all subsets $\actprof\subset\edges$ that contain exactly one outgoing edge for each $\play\in\vertices$.

We then consider the probability space $(\actions,\probalt_{\transitvec},2^{\actions})$, where
\begin{equation}
\label{eq:prob-Q}
\probalt_{\transitvec}(\actprof)=\weight_{\transitvec}(\actprof)=\prod_{(\play,\playalt)\in \actprof}\transit_{\play\playalt}.
\end{equation}
It is easy to see that $\probalt_{\transitvec}(\actions)=1$. From a probabilistic perspective, the multiplicative form of $\probalt_{\transitvec}$ implies that a random $\actprof\in\actions$ sampled according to $\probalt_{\transitvec}$, can be seen as the outcome of $\nPlayers$ independent 
draws of an outgoing edge $(\play,\playalt)\in\neighborsout_{\play}$ for each $\play\in\vertices$.

Let $\trees_{\!\!\play}(\vertices)$ denote the set of $\play$-rooted spanning trees in the complete graph with vertex set $\vertices$. 
Define 
\begin{equation}
\label{eq:treeweightdos}
\treeweight_{\play}(\transitvec):=\sum_{\tree\in\trees_{\!\!\play}(\vertices)}\weight_{\transitvec}(\tree)\quad\text{and}\quad
\treeweight_{\vertices}(\transitvec):=\sum_{\playalt\in\vertices}\treeweight_{\playalt}(\transitvec).
\end{equation}
Note that when computing $\treeweight_{\play}(\transitvec)$ it suffices to consider the spanning trees in the induced graph $\graph_{\transitvec}$, which contains only the edges with $\transit_{\play\playalt}>0$, since the remaining trees have weight zero. 
Notice also that by definition $\treeweight_{\vertices}>0$, since the graph $\graph_{\transitvec}$ has a unique strongly connected component. 
Moreover, a vertex $\play\in\vertices$ is transient iff $\treeweight_{\play}(\transitvec)=0$. 
In fact, the root of any rooted spanning tree of $\graph_{\transitvec}$ must lie in the strongly connected class. 
Hence, $\rc=\braces{\play\in\vertices:\treeweight_{\play}(\transitvec)>0}$.
However, to avoid keeping track of the dependence on the topology of $\graph_{\transitvec}$, it is convenient to consider all spanning trees in the complete graph with vertex set $\vertices$.

\begin{theorem}[Markov chain tree theorem]\label{th:MCTT}
	Consider a Markov chain with transition matrix $\transitvec$ and with a single recurrent class. 
	Then, the unique invariant measure $\stationary_{\transitvec}$ is given by 
	\begin{equation}\label{eq:stationary-Omega}
	\stationary_{\transitvec}(\play)=\frac{\treeweight_{\play}(\transitvec)}{\treeweight_{\vertices}(\transitvec)}.
	\end{equation}
\end{theorem}

As mentioned before, each $\actprof\in \actions$ can be identified with a pure strategy
profile in the deterministic buck-passing game, so the induced graph $\graph_{\actprof}$ is a 
disjoint union of unicycles. 
Let $\unicycles$ denote the set of all $\actprof\in\actions$ inducing a single spanning unicycle.
Moreover, let $\unicycles_{\play}$ be the spanning unicycles that have $\play$ in the cycle, 
and $\unicycles_{\play\playalt}$ those in which the edge $(\play,\playalt)$ is part of  the cycle.

\proof{Proof of \cref{th:MCTT}.}
	Call $\xomega_{\play}\coloneqq\treeweight_{\play}(\transitvec)$ and $\xomega\coloneqq\treeweight_{\vertices}(\transitvec)$.
	We observe that each $\actprof\in\unicycles_{\play\playalt}$ is of the form $\actprof=\tree\cup\braces{(\play,\playalt)}$ for a unique $\play$-rooted tree $\tree\in\trees_{\!\!\play}(\vertices)$ so that $\weight_{\transitvec}(\actprof)=\weight_{\transitvec}(\tree) \,\transit_{\play\playalt}$, 	and therefore $\probalt_{\transitvec}(\unicycles_{\play\playalt})=\xomega_{\play}\, \transit_{\play\playalt}$. 
	Now, the set $\unicycles_{\play}$  can be expressed as a 
	disjoint union $\unicycles_{\play}=\cupdot_{\playalt\in\vertices}\;\unicycles_{\play\playalt}$, so that
	\begin{equation}
	\probalt_{\transitvec}(\unicycles_{\play})=\sum_{\playalt\in\vertices}\probalt_{\transitvec}(\unicycles_{\play\playalt})=\sum_{\playalt\in \vertices}\xomega_{\play}\, \transit_{\play\playalt}=\xomega_{\play}.
	\end{equation}
	Similarly, if we focus on the edge $(\playaltalt,\play)$ preceeding $\play$, we may write 
	$\unicycles_{\play}=\cupdot_{\playaltalt\in \vertices}\;\unicycles_{\playaltalt\play}$, so that
	\begin{equation}
	\label{eq:x-i}
	\xomega_{\play}=\probalt_{\transitvec}(\unicycles_{\play})=\sum_{\playaltalt\in \vertices}\probalt_{\transitvec}(\unicycles_{\playaltalt\play})=\sum_{\playaltalt\in\vertices}\xomega_{\playaltalt} \,\transit_{\playaltalt\play}.
	\end{equation}
	This shows that $(\xomega_{\play})_{\play\in \vertices}$ is a left eigenvector of $\transitmatrix$ with eigenvalue $1$, so it is collinear with the invariant measure $\stationary_{\transitvec}$. The conclusion follows dividing each $\xomega_{\play}$ by $\xomega$.
\Halmos\endproof

As a by-product of the previous proof, we observe that $\treeweight_{\play}(\transitvec)$ can be expressed as the expected length of spanning unicycles.
Indeed, consider the random variable $\mathds{1}_{\braces{\actprof\in\unicycles_{\play}}}$, whose expected value is the probability that vertex $\play$ lies on the cycle of a spanning unicycle, that is,
\begin{equation}\label{eq:Es-in-U}
\expect_{\probalt_{\transitvec}}[\mathds{1}_{\unicycles_{\play}}]=\probalt_{\transitvec}(\unicycles_{\play})=\treeweight_{\play}(\transitvec).
\end{equation}
Moreover, let 
\begin{equation}
\label{eq:length}
\length(\actprof)\coloneqq\sum_{\play\in \vertices}\mathds{1}_{\braces{\actprof\in \unicycles_{\play}}}
\end{equation} 
be the length of the cycle if $\actprof\in\actions$ is a spanning unicycle, and $0$ otherwise. 

\begin{corollary}
	\label{co:expected-cycle-length} 
	Consider a Markov chain with transition matrix $\transitvec$ and with a single recurrent class $\rc\subseteq\vertices$. 
	Then 	$\treeweight_{\vertices}(\transitvec)={\expect}_{\probalt_{\transitvec}} [\length]$ and, in particular, 	$\treeweight_{\vertices}(\transitvec)\leq|\rc|$.
\end{corollary}

\begin{remark}
	Note that $\length(\actprof)$ is the total length of the cycle in the graph $\graph_{\actprof}$, which appears in the potential \eqref{eq:potentialdef} 
	for the \acl{DBPG}. In the next section, see \cref{eq:final-oprc}, we introduce a potential 
	for the \acl{SBPG} which involves the expected value ${\expect}_{\probalt_{\transitvec}} [\length]$.
\end{remark}

\section{The constrained buck-passing game.}
\label{se:gop}

We consider next a generalized version of the \acl{SBPG}, in which player $\play$'s strategy set is a subset $\sactions_{\play}\subset \simplex_{\play}$. 
Accordingly, we define $\sactions\coloneqq \times_{\play\in\vertices}\sactions_{\play}$. 
We  call this game $\game(\graph,\initial,\sactions)$ a \acfi{CBPG}\acused{CBPG}.
The set of its \aclp{NE} is denoted by $\NE(\sactions)$.
We will show that \acp{CBPG} are generalized ordinal potential games. 

For the sake of simplicity, consider first a strategy profile $\transitmatrix$ 
inducing an irreducible  Markov chain.  
In this case, by \cref{th:MCTT}, the cost for  player $\play$ is simply
\begin{equation}\label{eq:stationary-cost}
\cost_{\play}(\transitvec)=\stationary_{\transitvec}(\play)
=\frac{\treeweight_{\play}(\transitvec)}{\treeweight_{\vertices}(\transitvec)}.
\end{equation}
Since the numerator $\treeweight_{\play}(\transitvec)$ in \eqref{eq:stationary-cost} does not depend on $\transit_{\play}$, a profitable deviation 
for player $\play\in\vertices$ can only be achieved by increasing the denominator $\treeweight_{\vertices}(\transitvec)$. 
This suggests to take the map
$\potential(\transitmatrix)=-\treeweight_{\vertices}(\transitvec)$
as a generalized ordinal potential.
Since $\potential$ does not depend on $\initial$, any of its minimizers provides a \ac{PFNE}. 
This is indeed the case if every strategy profile $\transitvec\in\sactions$ 
gives rise to an irreducible Markov chain.

\subsection{A generalized ordinal potential.}
\label{suse:general-Markov-chains}

To get a workable expression for the costs in \cref{eq:cost-general}, we use the Markov chain tree formula.
To this end, consider the \emph{transient closures} of the recurrent classes 
\begin{equation}
\label{eq:transient-closure}
\tc_{\transitvec}^{\countrc}\coloneqq\braces*{\playalt\in\vertices:\absorb_{\transitmatrix}^{\playalt\to \countrc}=1}\quad\forall\,\countrc=1,\ldots,\numrc(\transitvec)
\end{equation}
and the \emph{residual transient class} that contains the remaining vertices
\begin{equation}
\label{eq:RTC}
\residual_{\transitvec} \coloneqq \braces*{\playalt\in\vertices : \absorb_{\transitvec}^{\playalt\to\countrc} < 1,\ \forall\countrc=1,\ldots,\numrc(\transitmatrix)},
\end{equation}
where $\absorb_{\transitmatrix}^{\playalt\to \countrc}$ is defined as in \eqref{eq:absorb}.

Each set $\tc_{\transitmatrix}^{\countrc}$ is closed with respect to the Markov chain and $\rc_{\transitmatrix}^{\countrc}\subseteq \tc_{\transitmatrix}^{\countrc}$.
Therefore, the restriction of the original Markov chain to $\tc_{\transitmatrix}^{\countrc}$ is itself a Markov chain having $\rc_{\transitmatrix}^{\countrc}$ as its unique recurrent class, so that \cref{th:MCTT} gives
\begin{equation}\label{eq:stationary_general}
\stationary_{\transitvec}^{\countrc}(\play)=
\begin{cases}
\dfrac{\treeweight_{\play}^{\countrc}(\transitvec)}{\treeweight^{\countrc}(\transitvec)}&\text{if }\play\in\tc^{\countrc}_{\transitvec}\\
0&\text{if }\play\in\vertices\setminus\tc^{\countrc}_{\transitvec}
\end{cases}
\end{equation}
where
\begin{equation}\label{eq:treeweightell}
\treeweight_{\play}^{\countrc}(\transitvec)=\sum_{\tree\in\trees_{\!\!\play}(\tc_{\transitvec}^{\countrc})}\weight_{\transitvec}(\tree)\quad \text{and}\quad \treeweight^{\countrc}(\transitvec)=\sum_{\play\in \tc_{\transitmatrix}^{\countrc}}\treeweight_{\play}^{\countrc}(\transitvec),
\end{equation}
with $\trees_{\!\!\play}(\tc_{\transitvec}^{\countrc})$  the set of $\play$-rooted  spanning trees  in the complete graph
over $\tc_{\transitvec}^{\countrc}$.

With these preliminaries, we may now state our main result for \acp{CBPG}.

\begin{theorem}
\label{th:cont-strong-exist}
Every \acl{CBPG} $\game(\graph,\initial,\sactions)$ admits the generalized ordinal potential
	\begin{equation}\label{eq:final-oprc}
	\potential(\transitvec):=\sum_{\countrc=1}^{\numrc(\transitvec)}(\nPlayers-\treeweight^{\countrc}(\transitvec)).
	\end{equation}
%	with $\treeweight^{\countrc}(\transitvec)$ defined as in  \eqref{eq:treeweightell}.
\end{theorem}

\proof{Proof.} Consider a profitable deviation by a player $\play$ from $\transitvec$ to $\transitvecalt\coloneqq(\transit'_{\play},\transitvec_{-\play})$, with $\cost_{\play}(\transitvecalt)<\cost_{\play}(\transitvec)$. This conveys the fact that $\cost_{\play}(\transitvec)>0$, so that from \cref{eq:cost-general} it follows that player $\play$ must belong to a recurrent class $\rc_{\transitvec}^{\countrcalt}$ with $\initial_{\transitvec}^{\countrcalt}>0$. 
	We distinguish two possible scenarios, depending on whether the player remains recurrent or becomes transient after deviating. 
	
	\smallskip	\noindent
	\textsc{Case 1:}   $\transit'_{\play\playalt}=0$ for all $\playalt\not\in\tc_{\transitvec}^{\countrcalt}$.
	In this case  $\play$ remains recurrent and, although $\rc_{\transitvec}^{\countrcalt}$ may change, 
	we have $\tc_{\transitvecalt}^{\countrcalt}=\tc_{\transitvec}^{\countrcalt}$ and 
	$\initial_{\transitvecalt}^{\countrcalt}=\initial_{\transitvec}^{\countrcalt}>0$.
	It then follows from \eqref{eq:cost-general} that $\cost_{\play}(\transitvecalt)<\cost_{\play}(\transitvec)$ is equivalent to
	$\stationary_{\transitvecalt}^{\countrcalt}(\play)<\stationary_{\transitvec}^{\countrcalt}(\play)$. 
	Now, since the weight of any $\play$-rooted tree
	$\tree\in\trees_{\!\!\play}(\tc_{\transitvec}^{\countrcalt})=\trees_{\!\!\play}(\tc_{\transitvecalt}^{\countrcalt})$ does not depend on $\transit_{\play}$ nor $\transit'_{\play}$, 
	we get $\treeweight_{\play}^{\countrcalt}(\transitvec)=\treeweight_{\play}^{\countrcalt}(\transitvecalt)$. 
	Then, from \eqref{eq:stationary_general} it follows that $\stationary_{\transitvecalt}^{\countrcalt}(\play)<\stationary_{\transitvec}^{\countrcalt}(\play)$
	is equivalent to $\treeweight^{\countrcalt}(\transitvecalt)>\treeweight^{\countrcalt}(\transitvec)$, which implies 
	$\potential(\transitvecalt)<\potential(\transitvec)$.

	\smallskip	\noindent
	\textsc{Case 2:}  $\transit'_{\play\playalt}>0$ for some $\playalt\not\in\tc_{\transitvec}^{\countrcalt}$.
	In this case all vertices in  $\tc_{\transitvec}^{\countrcalt}$, including $\play$, become transient and their costs drop to zero.
	We distinguish two subcases depending whether $\play$ becomes residual transient or it is absorbed into a different class.
	\\[1ex]
	\textsc{Case 2.1:} $\play\in\residual_{\transitvecalt}$.
	In this case the class $\tc^{\countrcalt}_{\transitvec}$  becomes part of $\residual_{\transitvecalt}$ and the remaining classes $\tc^{\countrc}_{\transitvec}$, $\countrc\neq\countrcalt$, remain unchanged.
	Hence, we lose the $\countrcalt$-th term in the sum of \eqref{eq:final-oprc}, and the other terms do not change. 
	From \cref{co:expected-cycle-length} we have $\treeweight^{\countrcalt}(\transitvec)\leq|\rc_{\transitvec}^{\countrcalt}|<\nPlayers$, so that the removed term is strictly positive and $\potential(\transitvecalt)<\potential(\transitvec)$.

	\smallskip	\noindent
	\textsc{Case 2.2:}  $\absorb_{\transitvecalt}^{\play\to \countrc}=1$ for some $\countrc\neq\countrcalt$.
	Here the full class  $\tc_{\transitvec}^{\countrcalt}$ is absorbed into the $\countrc$-th class, that is, $\tc^{\countrc}_{\transitvecalt}=\tc^{\countrc}_{\transitvec}\cup\tc^{\countrcalt}_{\transitvec}$, 	so that the $\countrcalt$-th and $\countrc$-th terms in the sum $\potential(\transitvec)$ are merged into the single 	$\countrc$-th term in $\potential(\transitvecalt)$ and the other terms do not change. 
	Hence, $\potential(\transitvecalt)<\potential(\transitvec)$ is equivalent to 
	\begin{equation*}
	\nPlayers-\treeweight^{\countrc}(\transitvecalt)<[\nPlayers-\treeweight^{\countrcalt}(\transitvec)]+[\nPlayers-\treeweight^{\countrc}(\transitvec)],
	\end{equation*}
	which follows by noting that $\treeweight^{\countrc}(\transitvecalt)>0$ and 
	using \cref{co:expected-cycle-length} once again, which gives
	\begin{equation*}
	\treeweight^{\countrcalt}(\transitvec)+\treeweight^{\countrc}(\transitvec)\leq|\rc^{\countrcalt}_{\transitvec}|+|\rc^{\countrc}_{\transitvec}|\leq\nPlayers.
	\end{equation*}
	
	In all scenarios we have that $\cost_{\play}(\transitvecalt)<\cost_{\play}(\transitvec)$ implies $\potential(\transitvecalt)<\potential(\transitvec)$,
	which proves that $\potential$ is a generalized ordinal potential.
\Halmos
\endproof

We stress the analogy between the potential function $\potential$ in \eqref{eq:final-oprc} and the one for the deterministic game in \eqref{eq:potentialdef}.
Indeed, the quantity $|\rc_{\actprof}^\countrc|$ in the latter is simply the number of rooted spanning trees for the $\countrc$-th unicycle and, since in the deterministic 
case the weight of each tree is $1$, we have $|\rc_{\actprof}^\countrc|=\treeweight^{\countrc}(\actprof)$. However, in contrast
with the deterministic case, $\potential$ may fail to provide an ordinal potential even when $\initial$ is fully supported.

\begin{figure}
	\FIGURE
	{
	\includegraphics[width=5cm]{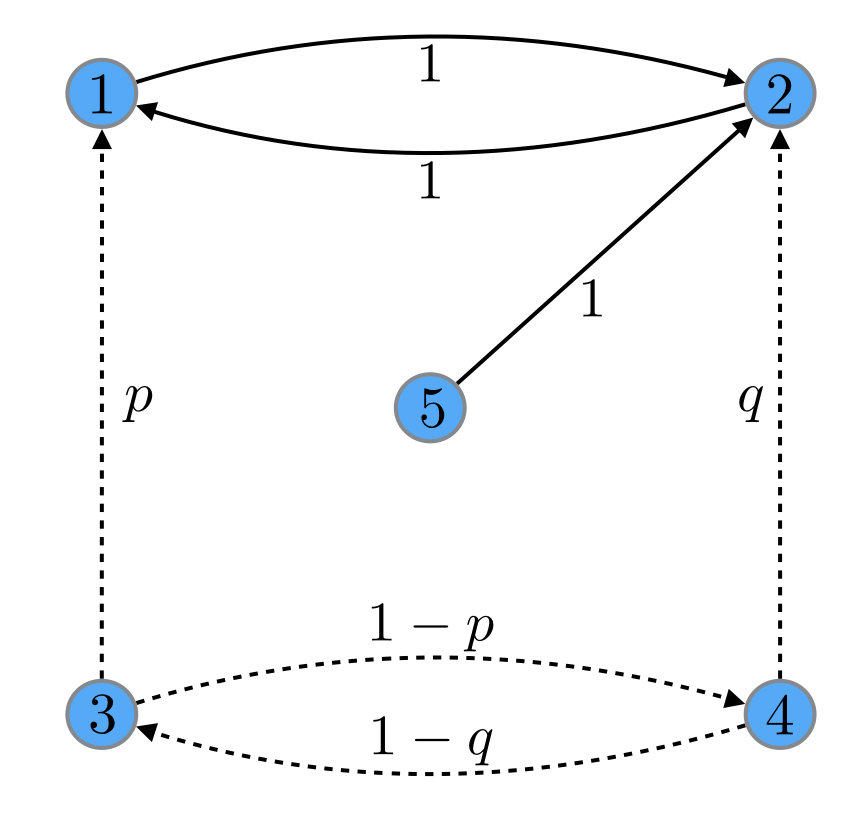}}
	{The strategy profile considered in \cref{ex:example-notLSC}.
	\label{fi:five-players}}
{}
\end{figure}
\begin{example}\label{ex:example-notLSC}
	Consider the graph in \cref{fi:five-players}. The set of players in the corresponding \ac{SBPG} is $\vertices=\braces{1,2,3,4,5}$ and the initial distribution $\initial$ is assumed to be uniform. 
	Fix  $p,q\in(0,1)$ and consider the following strategy profile 
	\begin{align*}
	&\transit_{12}=\transit_{21}=1,\\
	&\transit_{31}=\pex,\quad\transit_{34}=1-\pex,\\
	&\transit_{42}=\qex,\quad\transit_{43}=1-\qex,\\
	&\transit_{52}=1.
	\end{align*}	
	This strategy profile induces a unique recurrent class $\rc_{\transitvec}=\braces*{1,2}$ and we have
	\begin{equation*}
	\treeweight_{1}(\transitvec)=\treeweight_{2}(\transitvec)=\pex+\qex-\pex\qex, 	\end{equation*}
	so that  
	\begin{equation*}
	\potential(\transitvec)=5-2(\pex+\qex-\pex\qex).
	\end{equation*}
	Note that players $3$ and $4$ can decrease the potential by  increasing $\pex$ and $\qex$, respectively, although their cost remains $0$, since they are transient. 
	Therefore, $\potential$ is not an ordinal potential, even though $\initial$ is fully supported.
	Note also that $\potential(\transitvec) \to 5$,  as $\pex$ and $\qex$ tend to $0$, 
	whereas, for $\pex=\qex=0$, there are two recurrent classes  $\rc_{\transitvec}^{1}=\braces{1,2}$ and  $\rc_{\transitvec}^{2}=\braces{3,4}$ and the value of the potential is $6$. Therefore $\potential$ is not continuous, and 	not even lower semicontinuous.
\end{example}

%
% Subsection ------------------------------------------
%

\subsection{Buck-passing game and the Hamiltonian cycle problem.}
\label{suse:Hamiltonian}
	\cref{th:cont-strong-exist} can be connected to some literature that looks at the Hamiltonian cycle problem from the perspective of Markov chains. 
	That is,  considering the set of transition matrices that are compatible with a given graph $\graph$, this literature focuses on the class of functionals whose global minimum is attained on a permutation matrix which corresponds to a Hamiltonian cycle, provided it exists \citep[see, e.g.,][]{FilKra:MOR1994,Fil:FTSS2006,EjoLitNguTay:JAP2011,BorEjoFilNgu:Springer2012}.
	\cref{th:cont-strong-exist} shows that the potential function $\potential$ in \eqref{eq:final-oprc} belongs to this class.

\begin{corollary}\label{co:hamilton}
	Let the graph $\graph$ have a Hamiltonian cycle $\hamil$, and let $\transitvec_{\hamil}$ be
	the permutation matrix that represents this cycle. Then we have
	\begin{equation}
	\min_{\transitvec\in\simplex}\potential(\transitvec)=\potential(\transitvec_{\hamil})=0.
	\end{equation}
\end{corollary}

\proof{Proof.}
It follows immediately from the definition of $\potential$ in \eqref{eq:final-oprc} that $\potential(\transitvec_{\hamil})=0$. 
Indeed, in that case we have $\numrc(\transitvec_{\hamil})=1$, since the Hamiltonian cycle is strongly connected. 
Moreover, $\treeweight_{\play}(\transitvec_{\hamil}) = \nPlayers$, since for every vertex $\play$ there exists a unique tree rooted at $\play$ and this tree has unit weight. 
To conclude the proof it suffices to notice that, thanks to \cref{co:expected-cycle-length}, $\potential(\transitvec)\ge 0$ for all $\transitvec\in \simplex$. 
\Halmos
\endproof

\subsection{Existence of equilibria.}
\label{suse:equilibria}

We now address the existence of a \ac{PFNE} for general \aclp{CBPG}. 
The fact that $\game(\graph,\initial,\sactions)$ has a generalized ordinal potential guarantees the existence of $\varepsilon$-equilibria.
\begin{proposition}\label{pr:cont-weak-exist}
	For each $\varepsilon>0$ the  class of \aclp{CBPG} $\game(\graph,\cdot,\sactions)$  has  an \ac{eNE} which is also prior-free.
\end{proposition}

\proof{Proof.}
	The existence of an \ac{eNE} is a consequence of \citet[Lemmata~4.1 and 4.2]{MonSha:GEB1996}. See also \citet[Section 2.2.2.2]{LaCheSoo2016}.
	Prior-freeness is due to the fact that the potential function $\potential$ in \eqref{eq:final-oprc} does not depend on $\initial$. 
\Halmos
\endproof

Since  $\potential$ is not lower semicontinuous, even if  $\sactions$ is compact, we cannot invoke \cref{pr:GOP-pure_NE} to establish existence of equilibria, so we develop an {\em ad hoc} argument that requires some additional notation and preliminary results. 

Let $\sactions^{0}$ denote the set of strategy profiles $\transitvec\in\sactions$ with a minimal number of recurrent classes $\numrc^{0}=\min_{\transitvec\in\sactions}\numrc(\transitvec)$, and let $\sactions(\transitvec)$ be  the set of all unilateral deviations $\transitvecalt=(\transit'_{\play},\transitvec_{-i})$ by recurrent players $\play\in\rc_{\transitvec}^{1}\cupdot\cdots\cupdot\rc_{\transitvec}^{\numrc^{0}}$.
We start with the following simple observation.

\begin{lemma}\label{le:aux}
	For each $\transitvec\in\sactions^{0}$ and $\transitvecalt\in\sactions(\transitvec)$ we have $\numrc(\transitvecalt)=\numrc^{0}$ and $\tc^{\countrc}_{\transitvecalt}=\tc^{\countrc}_{\transitvec}$ for all $\countrc=1,\ldots,\numrc^{0}$. 
	Moreover, if  $\sactions$ is compact then $\inf_{\transitvecalt\in \sactions(\transitvec)} \potential(\transitvecalt)$ is attained.
\end{lemma}

\proof{Proof.}
	By the minimality of $\numrc(\transitvec)=\numrc^{0}$, a recurrent player $\play\in\rc_{\transitvec}^{\countrc}$ cannot become transient after a unilateral deviation $\transitvecalt=(\transit'_{\play},\transitvec_{-\play})$, so that $\transit'_{\play\playalt}=0$
	for all $\playalt\not\in\tc^{\countrc}_{\transitvec}$. 
	Hence, every such deviation preserves the number of classes $\numrc(\transitvecalt) =\numrc(\transitvec)=\numrc^{0}$, as well as 	all the transient closures $\tc^{\countrc}_{\transitvecalt}=\tc^{\countrc}_{\transitvec}$.
	It follows that, for all $\transitvecalt\in\sactions(\transitvec)$,
	\begin{equation}\label{eq:PsiOmega}
	\potential(\transitvecalt)=\sum_{\countrc=1}^{\numrc(\transitvec)}(\nPlayers-\treeweight^{\countrc}(\transitvecalt))
	\quad\text{with}\quad 
	\treeweight^{\countrc}(\transitvecalt)=\sum_{\play\in\tc^{\countrc}_{\transitvec}}\sum_{\tree\in\trees_{\!\!\play}(\tc^{\countrc}_{\transitvec})}\weight_{\transitvecalt}(\tree).
	\end{equation}
	For fixed $\transitvec$ these functions are continuous with respect to $\transitvecalt$ and the set $\sactions(\transitvec)$ is compact, since it is a section of a compact set. Therefore, the minimum of $\potential(\transitvecalt)$ over  $\sactions(\transitvec)$ is attained.
\Halmos
\endproof

Our next step is less trivial and requires the notion of \emph{skeleton} of a transient closure $\tc^{\countrc}_{\transitvec}$, defined as any rooted  tree 
\begin{equation}
\label{eq:skeleton}
\skeleton^{\countrc}_{\!\!\transitvec}\in \bigcup_{\play\in\tc^{\countrc}_{\transitvec}}\trees_{\!\!\play}(\tc^{\countrc}_{\transitvec})
\end{equation} 
having maximal weight $\weight_{\transitvec}(\tree)$.
Note that 
\begin{equation}\label{eq:skel}
\treeweight^{\countrc}(\transitvec)\leq \numroottrees_{\nPlayers}\,\weight_{\transitvec}(\skeleton^{\countrc}_{\!\!\transitvec}),
\end{equation}
where $\numroottrees_{\nPlayers}$ is the number of rooted trees on $\nPlayers$ vertices.
\begin{theorem}\label{co:existencePFNEpot}
	If $\sactions$ is compact, then there exists $\transitvec\in\sactions^{0}$ 
	such that $\potential(\transitvec)\leq\potential(\transitvecalt)$ for 
	all $\transitvecalt\in\sactions(\transitvec)$.
\end{theorem}

\proof{Proof.}
	Fix $\transitvec^{0}\in\sactions^{0}$ and, for $\countrc=1,\ldots,\numrc^{0}$, let $\tc^{\countrc}=\tc_{\transitvec^{0}}^{\countrc}$ be the corresponding transient closures.
	Consider a sequence defined inductively by
	\begin{equation}\label{eq:pi-k}
	\transitvec^{\run+1}\in\argmin_{\transitvecalt\in\sactions(\transitvec^\run)}\potential(\transitvecalt),
	\end{equation}
	so that $\numrc(\transitvec^\run)=\numrc^{0}$ and $\tc^{\countrc}_{\transitvec^\run}\equiv \tc^{\countrc}$ for $\countrc=1,\ldots,\numrc^{0}$
	and all $\run\in\naturals$.
	
	Since $\transitvec^\run\in\sactions(\transitvec^\run)$ we have $\potential(\transitvec^{\run+1})\leq \potential(\transitvec^{\run})$.
	If equality holds for some $\run$, then the conclusion follows by taking $\transitvec=\transitvec^\run$. 
	Consider then the case where $\potential(\transitvec^{\run+1})< \potential(\transitvec^{\run})$ for all $\run\in\naturals$. 
	Note that along the iterations we have
	\begin{equation}\label{eq:potalt}
	\potential(\transitvec^{\run}) 
	=\sum_{\countrc=1}^{\numrc^{0}}(\nPlayers-\treeweight^{\countrc}(\transitvec^\run))\quad\text{with}\quad 
	\treeweight^{\countrc}(\transitvec^\run)
	=\sum_{\play\in\tc^\countrc}\sum_{\tree\in\trees_{\!\!\play}(\tc^\countrc)}\weight_{\transitvec^\run}(\tree),
	\end{equation}
	so that  $\sum_{\countrc=1}^{\numrc^{0}}\treeweight^{\countrc}(\transitvec^\run)$ increases with $\run$.
	Moreover, $\transitvec^{\run+1}$ is obtained from $\transitvec^{\run}$ 
	by a deviation of a player $\play_{\run}$ in some recurrent class $\rc^{\countrc^\run}_{\transitvec^\run}$,
	so that only the $\countrc^\run$-th term in the sum changes and therefore $\treeweight^{\countrc}(\transitvec^\run)$ is nondecreasing in $\run$
	for each $\countrc=1,\ldots,\numrc^{0}$. 
	In particular $\treeweight^{\countrc}(\transitvec^\run)$ remains bounded away from $0$, and then, using \eqref{eq:skel}, we may find $\varepsilon>0$ such that,
	\begin{equation}\label{eq:skeletonbound}
	\text{for all }\run\in\naturals,\quad\weight_{\transitvec^\run}(\skeleton^{\countrc}_{\!\!\transitvec^\run})\geq\varepsilon.
	\end{equation}
	Take a convergent subsequence $\transitvec^{\run_{\subrun}}\to\transitvec\in\sactions$, and extract a further subsequence along which the skeletons are constant $\skeleton^{\countrc}_{\!\!\transitvec^{\run_{\subrun}}}\equiv \skeleton^{\countrc}$ for $\countrc=1,\ldots,\numrc^{0}$.
	Passing to the limit  in \eqref{eq:skeletonbound} along this subsequence we get $\weight_{\transitvec}(\skeleton^{\countrc})\geq\varepsilon$, which implies that $\tc^{\countrc}$ is still connected in the limit and therefore $\numrc(\transitvec)=\numrc^{0}$ and $\tc^{\countrc}_{\transitvec}=\tc^{\countrc}$. 
	From these facts, using \eqref{eq:potalt} and the continuity of the polynomials $\transitvec\mapsto\weight_{\transitvec}(\tree)$, we obtain $\potential(\transitvec^{\run_{\subrun}})\to \potential(\transitvec)$.
	Since $\potential(\transitvec^\run)$ is decreasing, we conclude in fact that the full sequence of potential values converges 	$\potential(\transitvec^\run)\to \potential(\transitvec)$.
	
	We now show that  the limit point $\transitvec$ satisfies the claim of the theorem.
	Indeed, we already proved that  $\numrc(\transitvec)=\numrc^{0}$, so that $\transitvec\in\sactions^{0}$. 
	Now, consider a player $\play\in\rc^{\countrc}_{\transitvec}$ and a deviation $\transitvecalt=(\transit'_{\play},\transitvec_{-\play})$. 
	Since $\play\in\rc^{\countrc}_{\transitvec}$, it follows that, for each $\playalt\in\tc^\countrc$, there is a path from $\playalt$ to $\play$ whose edges have positive probability under $\transitvec$.
	Since $\transitvec^{\run_{\subrun}}\to\transitvec$ this is also the case for
	$\transitvec^{\run_{\subrun}}$ for $\subrun$ large enough.
	Hence, $\play\in\rc_{\transitvec^{\run_{\subrun}}}^{\countrc}$ 	and then the definition of the sequence $\transitvec^\run$ implies that
	\begin{equation}\label{eq:potential-k-k_+1}
	\potential(\transitvec^{\run_{\subrun}+1})\leq \potential(\transit'_{\play},\transitvec_{-\play}^{\run_{\subrun}}).
	\end{equation}
	From \cref{le:aux} we have that $(\transit'_{\play},\transitvec^{\run_{\subrun}}_{-\play})$ has the
	same transient closures $\tc^{\countrc}$ as $\transitvec^{\run_{\subrun}}$, so we may write explicitly
	\begin{equation*}
	\potential(\transitvec^{\run_{\subrun}+1})\leq \potential(\transit'_{\play},\transitvec^{\run_{\subrun}}_{-\play})= \sum_{\countrc=1}^{\numrc^{0}}\parens*{\nPlayers-\sum_{\play\in\tc^\countrc}\sum_{\tree\in\trees_{\!\!\play}(\tc^{\countrc})}\weight_{(\transit'_{\play},\transitvec_{-\play}^{\run_{\subrun}})}(\tree)}.
	\end{equation*}
	Letting $\subrun\to\infty$, we conclude
	\begin{equation}\label{eq:potential-ineq}
	\potential(\transitvec)\leq \sum_{\countrc=1}^{\numrc^{0}}\parens*{\nPlayers-\sum_{\play\in\tc^\countrc}\sum_{\tree\in\trees_{\!\!\play}(\tc^{\countrc})}\weight_{(\transit'_{\play},\transitvec_{-\play})}(\tree)}=\potential(\transit'_{\play},\transitvec_{-\play}),
	\end{equation}
	where in the last equality  we used once again  \cref{le:aux},  according to which 
	$\numrc(\transit'_{\play},\transitvec_{-\play})=\numrc^{0}$ and 
	$\tc^{\countrc}_{(\transit'_{\play},\transitvec_{-\play})}=\tc^{\countrc}$.   This shows
	that $\potential(\transitvec)\leq\potential(\transitvecalt)$ for all $\transitvecalt\in\sactions(\transitvec)$, completing the proof.
\Halmos
\endproof

With these preliminaries, we may now prove the existence of prior-free equilibria.
\begin{theorem}\label{co:existencePFNE}
	Every class of \aclp{CBPG} $\game(\graph,\cdot,\sactions)$  in which $\sactions$ is compact admits a \ac{PFNE}.
\end{theorem}

\proof{Proof.}
	Consider $\transitvec\in\sactions^{0}$ as in \cref{co:existencePFNEpot}. 
	We will show that this $\transitvec$ is a \ac{NE} for every initial $\initial$. 
	Suppose by contradiction that there exists a player $\play$ and a deviation $\transitvecalt =(\transit'_{\play},\transitvec_{-\play})$ such that $\cost_{\play}(\transitvecalt)<\cost_{\play}(\transitvec)$.
	As  noted in the proof of \cref{th:cont-strong-exist}, player $\play$ must belong to some recurrent class $\rc^{\countrcalt}_{\transitvec}$ 	with $\initial^{\countrcalt}_{\transitvec}>0$.
	Moreover, $\transitvec\in\sactions^{0}$ so that  \cref{le:aux} implies that 	$\tc^{\countrcalt}_{\transitvecalt}=\tc^{\countrcalt}_{\transitvec}$ and, \emph{a fortiori}, $\initial^{\countrcalt}_{\transitvecalt}=\initial^{\countrcalt}_{\transitvec}$.
	Arguing as in {\sc Case 1}  in the proof of  \cref{th:cont-strong-exist}, the cost reduction must come from a decrease in the stationary probability $\stationary_{\transitvecalt}^{\countrcalt}(\play) <\stationary_{\transitvec}^{\countrcalt}(\play)$.
	This is in turn equivalent to $\treeweight^{\countrcalt}(\transitvecalt)>\treeweight^{\countrcalt}(\transitvec)$ and implies $\potential(\transitvecalt)<\potential(\transitvec)$, which contradicts the choice of $\transitvec$.
\Halmos\endproof

%
% Section ------------------------------------------
%

\section{The buck-holding game.}
\label{se:BHG}
In this section we consider a different game, called \acfi{BHG}\acused{BHG}, which is denoted by $\bhgame(\graph,\initial,\actions)$. 
This game is similar to the \acl{BPG} described in the previous sections, but now the goal of each player is to maximize the fraction of time in which she has the buck.
Hence, the cost in \eqref{eq:SBPG-cost} becomes a payoff.
The definitions of improvement and Nash equilibrium change accordingly.

\begin{definition} \label{de:nash-enash-payoff} 
	Consider a  game with payoffs $\parens{\cost_{\play}}_{\play\in\vertices}$.
	\begin{enumerate}[(a)]
		\item
		Given a strategy profile $\actprof\in\actions$, a \emph{unilateral deviation} for player $\play$ is a strategy $\actprofalt\in\actions$ which 
		differs from $\actprof$ only in its $\play$-th coordinate. It is  a \emph{profitable deviation} if in addition
		$\cost_{\play}(\actprofalt)>\cost_{\play}(\actprof)$, in which case the difference $\cost_{\play}(\actprofalt)-\cost_{\play}(\actprof)$ 
		is called the \emph{improvement} of player $\play$.
		
		\item
		A  strategy profile $\actprof\in\actions$ is a \acfi{NE}\acused{NE} if no player has a profitable deviation.
		Similarly, it is an \acfi{eNE}\acused{eNE} if no player has a profitable deviation with an improvement larger than $\varepsilon$.
	\end{enumerate}
\end{definition}

\cref{de:ordinal-potential} still holds, but now the goal is to maximize the potential.

In a \acfi{DBHG}\acused{DBHG} $\bhgame(\graph,\initial,\actions)$ each player chooses a single out-neighbor.

%The following definition is needed to analyze these games.
%\begin{definition}\label{de:wPFNE}
%	A strategy profile is called a \acfi{WPFNE}\acused{WPFNE} if it is an equilibrium for every  fully supported initial distribution $\initial$.
%\end{definition} 

\begin{proposition}\label{pr:BHG-gop}
Let $\bhgame(\graph,\cdot,\actions)$ any class of \aclp{DBHG}, then the following hold:

	\begin{enumerate}[\upshape(a\upshape)]
		\item\label{it:pr:BHG-gop-a} 
		For every initial distribution $\initial$, the game
	$\bhgame(\graph,\initial,\actions)$ is a generalized ordinal potential game, with the generalized ordinal potential $\potential$ as in \eqref{eq:potentialdef}. Moreover, if $\initial$ is fully supported on $\vertices$, then $\potential$ is in fact an ordinal potential.
	
		\item\label{it:pr:BHG-gop-b} 
		The class of games $\bhgame(\graph,\cdot,\actions)$ admits a \ac{PFNE}.
		
		\item\label{it:pr:BHG-gop-c} 
		Regardless of the choice of $\initial$, every improvement path has length $\bigoh(\nPlayers^{2})$.
		There exist instances with improvement paths of length  $\Theta(\nPlayers^2)$.
	\end{enumerate}
\end{proposition}

\proof{Proof.}
	Given a profile $\actprof\in\actions$, consider a player $\play\in\vertices$ who has a profitable deviation $\actalt_{\play}$
	and let $\actprofalt \coloneqq \parens{\actalt_{\play},\actprof_{-\play}}$.
	Since $\cost_{\play}(\actprofalt)>\cost_{\play}(\actprof)$, one of the following two scenarios  occurs:
	\begin{enumerate}[({\textsc{C}}$_1$)]
		\item\label{it:deviation-D1-BHG} 
		Player $\play\in\rc^{\countrcalt}_{\actprof}$ for some $\countrcalt\le\numrc(\actprof)$ and  $\actalt_{\play}\in \tc_{\actprof}^{\countrcalt}$ in such a way that the cycle where $\play$ lies becomes shorter. 
		Notice that in this case $\tc_{\actprof}^{\countrcalt}=\tc_{\actprofalt}^{\countrcalt}$, hence $\initial_{\actprof}^{\countrcalt}=\initial_{\actprofalt}^{\countrcalt}$ and the new payoff of player $\play$ is given by
		\begin{equation}\label{eq:new-payoff}
		\cost_{\play}(\actprofalt)=\frac{\abs{\rc_{\actprof}^{\countrcalt}}}{\abs{\rc_{\actprofalt}^{\countrcalt}}}\cost_{\play}(\actprof).
		\end{equation}
		Notice that $\actalt_{\play}$ can be a vertex in $\rc_{\actprof}^{\countrcalt}$ as well as a vertex in $\tc_{\actprof}^{\countrcalt}\setminus\rc_{\actprof}^{\countrcalt}$. Notice also that the deviation will be improving for player $\play$ if and only if $\initial^{\countrcalt}_{\actprof}>0$ and $|\rc_{\actprofalt}^{\countrcalt}|<|\rc_{\actprof}^{\countrcalt}|$.
		
		\item\label{it:deviation-D2-BHG}
		Player $\play$ is transient in $\actprof$ and becomes recurrent in $\actprofalt$, creating a new class. 
		In particular, assume that $\play\in\tc_{\actprof}^{\countrcalt}$ for some $\countrcalt\le\numrc(\actprof)$. 
Then
		\begin{equation}\label{eq:new-cycle}
		\rc_{\actprof}^{\countrc}=\rc_{\actprofalt}^{\countrc},\quad \text{for all }\countrc\le\numrc(\actprof)
		\end{equation}
		and $\actprofalt$ has a new cycle
		$\rc_{\actprofalt}^{\numrc(\actprofalt)}\ni\play$, with $\numrc(\actprofalt)=\numrc(\actprof)+1$.
		Moreover, 
		\begin{equation}\label{eq:eq-unicycle}
		\tc_{\actprof}^{\countrc}=\tc_{\actprofalt}^{\countrc},\quad \text{for all }\countrc\le\numrc(\actprof),\:\countrc\neq\countrcalt,
		\end{equation}
		and
		\begin{equation}\label{eq:new-unicycle}
		\tc_{\actprofalt}^{\numrc(\actprofalt)}\cup\tc_{\actprofalt}^{\countrcalt}=\tc_{\actprof}^{\countrcalt}.
		\end{equation}
		Notice that also in this case the deviation is profitable for $\play$ if and only if $\initial^{\numrc(\actprofalt)}_{\actprofalt}>0$. 
	\end{enumerate}
	
	The claims of the theorem now follow straightforwardly.
	
	\noindent
	\ref{it:pr:BHG-gop-a}
	We argue as in \cref{th:GOP}. 
	Let $\potential$ be defined as in \eqref{eq:potentialdef}.
	If $\cost_{\play}(\actprofalt)>\cost_{\play}(\actprof)$, then $\potential(\actprofalt)>\potential(\actprof)$, both under \ref{it:deviation-D1-BHG} and \ref{it:deviation-D2-BHG}.
	Indeed, if a deviation of type \ref{it:deviation-D1-BHG} takes place, then
	\begin{equation}\label{eq:potential-C1}
	\potential(\actprofalt)-\potential(\actprof)
	=\abs{\rc_{\actprof}^{\countrc}}-\abs{\rc_{\actprofalt}^{\countrc}}>0.
	\end{equation}
	On the other hand, if a deviation of type \ref{it:deviation-D2-BHG} takes place, then
	\begin{equation}\label{eq:potential-C2}
	\potential(\actprofalt)-\potential(\actprof)
	=\nPlayers-\abs{\rc_{\actprofalt}^{\countrcalt}} + \nPlayers-\abs{\rc_{\actprofalt}^{\numrc(\actprofalt)}} - (\nPlayers-\abs{\rc_{\actprof}^{\countrcalt}})>0.
	\end{equation}
	This proves that the game is generalized ordinal potential. 
	To prove that under the assumption that $\initial$ is fully supported the game is, indeed, ordinal potential, we have to check that the assumptions $\potential(\actprofalt)>\potential(\actprof)$ and $\initial(\play)>0$ for all $\play\in\vertices$ imply a profitable deviation of type \ref{it:deviation-D1-BHG} or \ref{it:deviation-D2-BHG}. 
	This, again, follows the line of \cref{th:GOP}. 
	Indeed, if $\actprofalt=(\actalt_{\play},\actprof_{-\play})$ and $\potential(\actprofalt)>\potential(\actprof)$, then, either $\actprofalt$ has one extra cycle, which implies that the deviation $\actprof$ to $\actprofalt$ is of type \ref{it:deviation-D2-BHG} and, by the fact that $\initial(\play)>0$, is improving for the deviating player, or the cycle in which $\play$ lies in $\actprof$ has been shortened in $\actprofalt$, which implies that the deviation is of type \ref{it:deviation-D1-BHG}, therefore $\cost_{\play}(\actprofalt)>\cost_{\play}(\actprof)$.

	\noindent
	\ref{it:pr:BHG-gop-b}
	To prove the existence of \acp{PFNE},  notice that the strategies in which $\potential$ is maximized are \acp{NE} for every initial distribution $\initial$.
	%It is worth noting that, since we can  claim  the potential nature of the game only if $\initial$ is fully supported, weak prior-freeness is the best we can achieve with the techniques developed in the previous sections.

	\noindent
	\ref{it:pr:BHG-gop-c}
	To show that the uniform upper bound of \cref{th:DBPG-FIP} applies also to \acp{DBHG}, consider the following: 
	Deviations of type \ref{it:deviation-D2-BHG} can occur at most $\floor{\nPlayers/2}-1$ times, since  the number of cycles is between $1$ and $\floor{\nPlayers/2}$. On the other hand, each cycle can be shrunk at most $\nPlayers-2$ times. 
	Hence, we have a quadratic upper bound. 
	On the other hand, a lower bound of the same order of magnitude can be obtained using a complete graph with an even number of vertices, as in the proof of \cref{th:DBPG-FIP}.
	The starting configuration is now a Hamiltonian cycle and the improvement steps are the same of the \ac{DBPG}, but in reverse order.
\Halmos
\endproof

\begin{remark}
	We say that $\graph$ admits a perfect matching if there exists a partition $\vertices^{1}\cupdot\vertices^{2}=\vertices$ and a subset $\widetilde{\edges}\subset\edges$ of cardinality $\nPlayers$ such that, for each $\play\in\vertices^{1}$ and $\playalt\in\vertices^{2}$, both $\parens{\play,\playalt}\in\widetilde{\edges}$ and $\parens{\playalt,\play}\in\widetilde{\edges}$.
	As mentioned in \cref{re:Hamilton}, in a \ac{DBPG}, if $\graph$ admits a Hamiltonian cycle, then the strategy profile $\actprof$ in which players play along such cycle is a \ac{NE}.
	In a \ac{DBHG}, if  $\graph$ admits a perfect matching, then the strategy profile that realizes this matching is a \ac{NE}.
	Indeed, in this case, for any possible deviation, the payoff of the deviating player would drop to zero. 
	More generally, we saw in \cref{se:deterministic} that for a connected graph $\graph$, every unicyclic strategy in which the cycle cannot be extended by a unilateral deviation is a \ac{PFNE}. 
	Similarly, in a \ac{DBHG}, a subgraph $\actprof\in\actions$ in which every player is in a cycle that she cannot unilaterally shorten is a \ac{PFNE}. 
\end{remark}

We now show the \emph{holding analogue} of \cref{th:cont-strong-exist}. Given a directed graph $\graph(\vertices,\edges)$, an initial distribution $\initial$, and an arbitrary set of strategy profiles $\sactions$ as in \cref{se:gop}, we consider the \acl{CBHG} $\bhgame(\graph,\initial,\sactions)$. 

\begin{proposition}\label{pr:CBHG}
	Every \acl{CBHG} $\bhgame(\graph,\initial,\sactions)$ is generalized ordinal potential with generalized ordinal potential function $\potential$ as in  \eqref{eq:potentialdef}.
\end{proposition}

In what follows, we are going to use the same notation as in the proof of \cref{pr:BHG}.

\proof{Proof.}
	Consider a profitable deviation for player $\play$ from $\transitvec$ to $\transitvecalt\coloneqq(\transit'_{\play},\transitvec_{-\play})$, with $\cost_{\play}(\transitvecalt)>\cost_{\play}(\transitvec)$. As for the deterministic case, we distinguish two possible scenarios, depending on whether player $\play$ is recurrent or transient before the deviation.
	Notice that player $\play$ cannot improve her payoff by moving some probability mass 
	out of her own transient closure class.
	
	\smallskip\noindent
	\textsc{Case 1:} $\play$ is recurrent both in $\transitvec$ and   $\transitvecalt$. Call $\rc^{\countrcalt}_{\transitvec}$ its recurrent class under the strategy $\transitvec$.
	This is the stochastic version of \ref{it:deviation-D1-BHG}.
	Although $\rc_{\transitvec}^{\countrcalt}$ may not coincide with $\rc_{\transitvecalt}^{\countrcalt}$ , 
	we have $\tc_{\transitvecalt}^{\countrcalt}=\tc_{\transitvec}^{\countrcalt}$ and 
	$\initial_{\transitvecalt}^{\countrcalt}=\initial_{\transitvec}^{\countrcalt}$.
	It then follows from \eqref{eq:cost-general} that $\cost_{\play}(\transitvecalt)>\cost_{\play}(\transitvec)$ is equivalent to having $\initial_{\transitvecalt}^{\countrcalt}=\initial_{\transitvec}^{\countrcalt}>0$ and
	$\stationary_{\transitvecalt}^{\countrcalt}(\play)>\stationary_{\transitvec}^{\countrcalt}(\play)$. 
	Now, since the weight of any $\play$-rooted tree
	$\tree\in\trees_{\!\!\play}(\tc_{\transitvec}^{\countrcalt})=\trees_{\!\!\play}(\tc_{\transitvecalt}^{\countrcalt})$  depends neither on $\transit_{\play}$ nor on $\transit'_{\play}$, 
	we get $\treeweight_{\play}^{\countrcalt}(\transitvec)=\treeweight_{\play}^{\countrcalt}(\transitvecalt)$. 
	Then, from \eqref{eq:stationary_general} it follows that $\stationary_{\transitvecalt}^{\countrcalt}(\play)>\stationary_{\transitvec}^{\countrcalt}(\play)$
	is equivalent to $\treeweight^{\countrcalt}(\transitvecalt)<\treeweight^{\countrcalt}(\transitvec)$, which implies 
	$\potential(\transitvecalt) > \potential(\transitvec)$.

	\smallskip\noindent
	\textsc{Case 2:} $\play$ transient in $\transitvec$  but recurrent in $\transitvecalt$. In this case $\play$ creates a new recurrent class.
	This is the stochastic version of \ref{it:deviation-D2-BHG}. We call $\tc^{\countrcalt}_{\transitvec}$ the transient closed class of $\play$ under $\transitvec$. Moreover, we use the same notation as in \cref{eq:new-cycle,eq:eq-unicycle,eq:new-unicycle}.
	We have, %using the notation of \cref{pr:BHG-gop}, we have 
	\begin{align*}
	\potential(\transitvecalt)-\potential(\transitvec)
	&=\nPlayers-\treeweight^{\countrcalt}(\transitvecalt) + \nPlayers-\treeweight^{\numrc(\transitvecalt)}(\transitvecalt) -(\nPlayers-\treeweight^{\countrcalt}(\transitvec))\\
	&=\nPlayers+\treeweight^{\countrcalt}(\transitvec)-\treeweight^{\numrc(\transitvecalt)}(\transitvecalt)-\treeweight^{\countrcalt}(\transitvecalt)\\
	&\ge\nPlayers+\treeweight^{\countrcalt}(\transitvec)
	-\abs{\tc^{\numrc(\transitvecalt)}_{\transitvecalt}} - \abs{\tc^{\countrcalt}_{\transitvecalt}}\\
	&=\nPlayers+\treeweight^{\countrcalt}(\transitvec)
	-\abs{\tc^{\countrcalt}_{\transitvec}} > 0,
	\end{align*}
	where the  inequality stems from \cref{co:expected-cycle-length}, $\treeweight^{\countrcalt}(\transitvec)>0$ and $\abs{\tc^{\countrcalt}_{\transitvec}}\le\nPlayers$.
\Halmos
\endproof

\begin{corollary}\label{co:BHG-epsilon}
	For every  $\varepsilon>0$, every class of \aclp{CBHG} $\bhgame(\graph,\cdot,\sactions)$ has the \ac{eFIP} and admits a prior-free \ac{eNE}.
\end{corollary}

Unfortunately, mimicking the argument of \cref{co:existencePFNEpot,co:existencePFNE,le:aux} is not enough to prove the existence of \ac{PFNE} for general compact strategy space $\sactions$. Indeed, the notion of \emph{skeleton} introduced in \cref{suse:equilibria} does not guarantee that the transient closures are retained in the limit.
Nonetheless, the following proposition easily follows by the previous analysis.

\begin{proposition}\label{pr:BHG}
	Every class of \acl{CBHG} $\bhgame(\graph,\cdot,\sactions)$ admits a \ac{PFNE} if one of the following holds:
	\begin{enumerate}[\upshape(i\upshape)]
		\item 
		\label{it:pr:BHG-1}
		The set $\sactions$ is finite.
		\item 
		\label{it:pr:BHG-2}
For every $\transitvec\in\sactions$ the associated Markov chain has a unique recurrent class.
	\end{enumerate}
\end{proposition}

We remark that the second condition in  \cref{pr:BHG} is immediately satisfied when 
\begin{equation}
\forall\transitvec\in\sactions\qquad\transit_{\play,\playalt}>0,\quad\forall \play\neq \playalt.
\end{equation}
In what follows we will see that the latter condition is satisfied by the PageRank game mentioned in the introduction.

\subsection{The PageRank game.}
\label{suse:BHG-pagerank}
The PageRank dynamics was introduced by \citet{BrinPage:1998} as a tool to rank webpages.
From a mathematical perspective, PageRank is a Markov chain on the state space $\vertices$ of web-pages, where two webpages $\play,\playalt\in\vertices$ are connected by a directed edge $(\play,\playalt)\in\edges$ if there exists a weblink on page $\play$ leading to page $\playalt$. 
A websurfer visiting a given page $\play$ at time $\per$ clicks at random on a link $(\play,\playalt)\in\edges$  and  moves to page $\playalt$ at time $\per+1$. 
Alternatively, with small  probability, she chooses one of the billion webpages in $\vertices$ according to some  distribution.
This describes a Markov chain having a unique stationary measure  $\stationary$, according to which webpages are then ranked.

We now consider a game-theoretic version of this problem, where players are webmasters whose strategies are the out-links of their webpages and the payoff is the ranking of their pages.  
This model lies in the realm of  \acp{BHG}. 
In particular, the game can be framed as follows: we identify the web pages with the set $\vertices=\braces{1,\dots,\nPlayers}$ and we consider $\graph=\completegraph_{\nPlayers}$, the complete graph. 
We fix a dumping factor $\dumping\in(0,1)$ and a probability measure  on $\vertices$, which we identify with a nonnegative column vector $\PRvec$ of size $\nPlayers$. 
For each player $\play$, the strategy set $\PRactions_{\play}\subset 2^{\vertices\setminus\braces{\play}}$, i.e., a strategy $\PRact_{\play}$ of player $\play$ is a subset of $\vertices\setminus\braces{\play}$ that satisfies some constraints. 
For instance, a page cannot connect to more than a fixed number of other pages; alternatively, if the page is about some topic it must link to at least another page with a related content, etc.

Given a strategy profile $\PRactprof=\parens{\PRact_{1},\dots,\PRact_{\nPlayers}}\in\PRactions$, define the transition matrix $\PRtrans$ with entries
\begin{equation}\label{eq:transPR}
\PRtrans_{\play\playalt} \coloneqq \frac{\mathds{1}_{\playalt\in\PRact_{\play}}}{\abs{\PRact_{\play}}}.
\end{equation}
According to the transition  matrix $\PRtrans$ player $\play$  chooses uniformly at random one of the players in $\PRact_{\play}$. We then consider the perturbation given by
\begin{equation}\label{eq:trans-perturbed}
\transitmatrix \coloneqq (1-\dumping)\PRtrans+\dumping\mathbf{1}\PRvec^{\top},
\end{equation}
where $\mathbf{1}$ is a column vector whose components are all $1$.
Notice that, since $\dumping\in(0,1)$, regardless of the particular choice of $\PRvec$, the transition matrix $\transitvec$ has a unique recurrent class. 
Hence, it has a unique stationary measure, which is called PageRank. 
In particular, if $\PRvec$ is strictly positive, then $\transitvec$ is irreducible, and all the entries of the PageRank are positive.

The above game can be framed as a \acl{CBHG} as follows:
The  set $\PRactions_{\play}$ can be mapped to the strategy set $\sactions_{\play}$ of vectors $\transitvec_{\play}$ such that 
\begin{equation}\label{eq:PRtransit-i}
\transit_{\play\playalt} = (1-\dumping)\frac{\mathds{1}_{\playalt\in\PRact_{i}}}{\abs{\PRact_{\play}}} + \dumping\PRstr_{\playalt}.
\end{equation}
Since every $\transitmatrix\in\sactions$ has a unique stationary distribution, the payoff vector does not depend on the initial measure $\initial$.
Therefore, \cref{pr:BHG}\ref{it:pr:BHG-2} applies and all the equilibria of this game are prior-free.
%For the relation between \cref{pr:BHG}\ref{it:pr:BHG-2} and  regular/singular perturbations of Markov chains we refer the reader to \citep{AvrLit:INRIA2004,AvrLitSon:IM2008,AvrFilHow:SIAM2013}.

As mentioned in the Introduction, this game has been introduced for the first time in 
\citet{HopShe:mimeo2008}.
In our notation, their model coincides with the particular case in which
\begin{equation}
\label{eq:hs1}
\PRactions_{\play}=2^{\vertices\setminus\{\play\}},\qquad\forall\play\in\vertices,
\end{equation} 
that is, each webpage can link to any other webpage, without any constraints.
%and
%\begin{equation}\label{hp:hs2}
%\PRvec_{\play},\qquad\forall \play\in\vertices.
%\end{equation}
\citet{HopShe:mimeo2008} show that in their game equilibria exist and  establish some features that all of them share. 
In a subsequent work \citet{CheTenWanZho:FA2009} prove that, under the same set of assumption, there exists equilibria which are sensitive to the choice of the parameter $\dumping$. 
Our work adds to this literature by weakening the assumptions for the existence of equilibria and reveiling the potential nature of the game.

%\citet{AviIwaPak:DAM2014} consider versions of the PageRank game on undirected networks, where players cannot unilaterally create links, but they can delete existing links.
%\citet{KouMarPapRigSid:SAGT2015} deal with a game of the PageRank type where players choose both their outgoing links and their weights. 
%Using quasi-concavity arguments they prove existence of pure Nash equilibria.

% Acknowledgments here
\section*{Acknowledgments.}
% Enter the text of acknowledgments here
The authors thank the two anonymous referee and the area editor for their insightful comments and for pointing out several relevant references. 
Matteo Quattropani gratefully thanks Pietro Caputo for several interesting discussions.
Matteo Quattropani and Marco Scarsini are a members of INdAM-GNAMPA.
Roberto Cominetti gratefully acknowledges the support of LUISS during a visit in which this research was initiated, as well as the support of the Complex Engineering Systems Institute, ISCI (ICM-FIC: P05-004-F, CONICYT: FB0816)
and FONDECYT 1171501.
Marco Scarsini gratefully acknowledges the support and hospitality of FONDECYT 1130564 and N\'ucleo Milenio ``Informaci\'on y Coordinaci\'on en Redes.''
This research project received partial support from the COST action GAMENET, the INdAM-GNAMPA Project 2020 ``Random walks on random games,'' and  the Italian MIUR PRIN 2017 Project ALGADIMAR ``Algorithms, Games, and Digital Markets.'' 

% References here (outcomment the appropriate case) 

% CASE 1: BiBTeX used to constantly update the references 
%   (while the paper is being written).
\bibliographystyle{informs2014} % outcomment this and next line in Case 1
\bibliography{bibmarkov} % if more than one, comma separated

% CASE 2: BiBTeX used to generate mypaper.bbl (to be further fine tuned)
%\input{mypaper.bbl} % outcomment this line in Case 2
%\bibliography{bibtex/bibmarkov}

\newpage

\begin{APPENDIX}{List of symbols}
%	\MQ{Check the list of symbols}
%\section{List of symbols}

\label{se:symbols}

The following table contains the symbols that we have used throughout the paper.

\begin{longtable}{p{.10\textwidth} p{.85\textwidth}}

$\tc_{\actprof}^{\countrc}$ & unicycle induced by $\actprof$, defined in  \cref{de:InducedGraph,eq:transient-closure} \\

%$\adj$ adjugate matrix, defined in \cref{eq:adjugate}

$\PRact_{\play}$ &  strategy of player $\play$  in the PageRank game  \\

$\PRactprof$ &  strategy profile  in the PageRank game  \\

$\PRactions_{\play}$ &  strategy set of player $\play$ in the PageRank game  \\

$\PRactions$ &  set of strategy profiles  in the PageRank game  \\

$\cost_{\play}$ &  cost function of player $\play$, defined in \cref{eq:SBPG-cost,eq:cost-general} \\

$\cyclegraph_{\nPlayers}$ &  bi-directional cycle  \\

$\rc_{\actprof}^{\countrc}$ &  cycle induced by $\actprof$, defined in \cref{de:InducedGraph} \\

$\edges$ &  set of edges \\

$\edges_{\actprof}$ &  subset of edges induced by $\actprof$, defined in \cref{de:InducedGraph} \\

$\edges_{\transitvec}$ &  set of weighted edges with weights determined by $\transitvec$ \\

$\graph$ &  directed graph \\

$\graph_{\actprof}$ &  subgraph induced by $\actprof$, defined in \cref{de:InducedGraph} \\

$\graph_{\transitvec}$ &  weighted graph induced by $\transitvec$ \\

$\hamil$ &  Hamiltonian cycle \\

$\completegraph_{\nPlayers}$ &  complete graph  \\

$\countrc(\play)$ &  label of the unicycle that contains player $\play$ \\

%$\laplace_{\transitvec}$ Laplacian

$\halfplay$ & $\nPlayers/2-1$  \\

$\numrc(\actprof)$ &  number of unicycles under $\actprof$, defined in \cref{de:InducedGraph} \\

$\numrc^{0}$ & $\min_{\transitvec\in\sactions}\numrc(\transitvec)$  \\

$\nPlayers$ &  number of players \\

$\neighborsout_{\play}$ &  out-neighbors of player $\play$ \\

$\NE(\actions)$ &  Nash equilibria in \ac{DBPG} \\

$\NE(\simplex)$ &  Nash equilibria in \ac{SBPG} \\

$\NE(\sactions)$ &  Nash equilibria in \ac{CBPG} \\

$\prob_{\transitvec}$ &  probability measure induced by $\transitvec$, defined in \eqref{eq:Markov-chain} \\

$\absorb_{\transitvec}^{\countrc}$ &  probability of absorption in  $\rc_{\transitmatrix}^{\countrc}$  \\

$\absorb_{\transitmatrix}^{\playalt\to \countrc}$ &  probability of absorption in  $\rc_{\transitmatrix}^{\countrc}$ starting from $\playalt$, defined in \eqref{eq:absorb} \\

$\PoA$ &  price of anarchy, defined in \cref{eq:PoA-D,eq:PoA-S} \\

$\PoS$ &  price of stability, defined in \cref{eq:PoS-D,eq:PoS-S} \\

$\probalt_{\transitvec}$ &  restriction of $\weight_{\transitvec}$ to $\actions$, defined in \eqref{eq:prob-Q} \\

$\PRtrans$ &  transition matrix in the PageRank game, defined in \eqref{eq:transPR} \\

$\numroottrees_{\nPlayers}$ &  number of rooted trees on $\nPlayers$ vertices \\

$\residual_{\transitvec}$ &  residual transient class, defined in \eqref{eq:RTC} \\

$\act_{\play}$ &  strategy of player $\play$  in \ac{DBPG} \\

$\actprof$ &  strategy profile in \ac{DBPG} \\

$\actions_{\play}$ &  strategy set of player $\play$ in \ac{DBPG} \\

$\actions$ &  set of strategy profiles in \ac{DBPG} \\

$\SC$ &  social cost function, defined in \eqref{eq:SC} \\

$\per$ &  time \\

$\horizon_{\play}$ &  hitting time of $\play$ \\

$\trees(\play)$ &  set of $\play$-rooted spanning trees \\

$\unicycles$ &  set of all $\actprof\in\actions$ inducing a single spanning unicycle \\

$\unicycles_{\play}$ &  spanning unicycles that have $\play$ in the cycle \\

$\unicycles_{\play\playalt}$ &  spanning unicycles in which the edge $(\play,\playalt)$ is part of  the cycle \\

$\vertices$ &  set of vertices \\

$\vertices_{\transitmatrix}^{0}$ &  set of transient vertices \\

$\ringstargraph_{\nPlayers}$ &  wheel graph  \\

$\xomega_{\play}$ & $\probalt_{\transitvec}(\unicycles_{\play})$, defined in \eqref{eq:x-i}  \\

$\markov_{\per}$ &  Markov chain \\

$\dumping$ &  dumping factor  in the PageRank game \\

$\game(\graph,\initial,\actions)$ &   \acl{DBPG}  \\

$\game(\graph,\initial,\simplex)$ &  \acl{SBPG} \\

$\game(\graph,\initial,\sactions)$ & \acl{CBPG} \\

$\bhgame(\graph,\initial,\actions)$ & \acl{DBHG} \\

$\indunic_{\play}(\actprof)$ &  indicator of $\play\in\rc_{\actprof}^{\countrc(\play)}$, defined in \eqref{eq:ind-unicycle} \\

$\buck_{\play,\per}(\actprof)$ &  indicator of the event that player $\play$ has the buck at time $\per$ under profile $\actprof$, defined in \eqref{eq:buck} \\

%$\eigen_{\transitmatrix}^{(\run)}$ eigenvalue of the Laplacian

$\length(\actprof)$ &  length of the cycle if $\actprof\in\actions$ is a spanning unicycle, and $0$ otherwise, defined in \eqref{eq:length} \\

$\initial$ &  initial measure \\

$\initial_{\actprof}^{\countrc}$ &  probability that the buck is  assigned initially to a vertex in $\tc_{\actprof}^{\countrc}$, defined in \eqref{eq:initial-ell} \\

$\initial_{\actprof}^{\countrc}$ &  probability that the buck is absorbed in $\rc_{\actprof}^{\countrc}$, defined in \eqref{eq:total-mass} \\

$\PRvec$ &  fully supported probability measure  in the PageRank game  \\

$\sactions_{\play}$ &  strategy set of player $\play$ in \ac{CBPG} \\

$\sactions$ &  set of strategy profiles in \ac{CBPG} \\

$\sactions^{0}$ &  set of strategy profiles $\transitvec\in\sactions$ with a minimal number of recurrent classes \\

$\transitvec_{\play}$ &  strategy of player $\play$  in \ac{SBPG} \\

$\transitvec$ &  strategy profile in \ac{SBPG} \\

$\transitmatrix$ &  transition matrix \\

$\stationary_{\transitmatrix}$ &  stationary measure induced by $\transitmatrix$ \\

$\stationary_{\transitmatrix}^{\countrc}$ &  stationary measure induced by $\transitmatrix$ on the class $\rc_{\transitmatrix}^{\countrc}$ \\

$\stationary_{\transitvec}^{\countrc}$ &  stationary measure on $\rc_{\transitmatrix}^{\countrc}$  \\

$\simplex$ &  set of strategy profiles in \ac{SBPG} \\

$\simplex_{\play}$ &  strategy set of player $\play$ in \ac{SBPG} \\

$\tree$ &  tree \\

$\skeleton^{\countrc}_{\!\!\transitvec}$ &  skeleton of $\tc^{\countrc}_{\transitvec}$, defined in \eqref{eq:skeleton} \\

$\potential$ &  potential function, defined in \cref{de:ordinal-potential,eq:potentialdef} \\

$\weight_{\transitvec}$ &  weight function, defined in \eqref{eq:weight} \\

$\treeweight_{\play}(\transitvec)$ &  tree-volume of vertex $\play$, defined in \eqref{eq:treeweightdos} \\

$\treeweight_{\vertices}(\transitvec)$ &  tree-volume of the Markov chain, defined in \eqref{eq:treeweightdos} \\

$\abs{\event}$ &  cardinality of set $\event$ \\

$\mathds{1}_{\event}$ &  indicator of set $\event$ \\

\end{longtable}
 \end{APPENDIX}

\end{document}